\documentclass{article}
\usepackage{graphicx} % Required for inserting images
\usepackage[a4paper, total={6in, 8in}]{geometry}
\usepackage{authblk}
\usepackage{physics}
\usepackage{mathrsfs}
\usepackage{times}

\newtheorem{theorem}{Theorem}[section]
\newtheorem{corollary}[theorem]{Corollary}
\newtheorem{definition}[theorem]{Definition}

\newtheorem{lemma}[theorem]{Lemma}
\newtheorem{conjecture}[theorem]{Conjecture}
\newtheorem{proposition}[theorem]{Proposition}
\newtheorem{remark}[theorem]{Remark}

\newenvironment{proof}[1][Proof]{\paragraph{#1:}}{\hfill$\square$}
\usepackage{amsmath} %serve per alcune funzioni
\usepackage{amssymb}
\usepackage{xcolor}
\usepackage{hyperref}%serve per far si che i label-ref rimandino ai pezzi quando ci si preme sopra
\usepackage{url,changes}

\usepackage{cleveref}
\usepackage{csquotes}
\usepackage[%
    style=phys,%
    articletitle=true,biblabel=brackets,%
    chaptertitle=false,pageranges=false%
    ]{biblatex}
\usepackage{ytableau}

\hypersetup{
    colorlinks=true,
    linkcolor=[rgb]{0.3, 0.6, 0.3},
    filecolor=[rgb]{0.3, 0.6, 0.3},
    citecolor=[rgb]{0.3, 0.6, 0.3},
    urlcolor=[rgb]{0.3, 0.6, 0.3},
}

\begin{document}

\title{Robust Self-Testing of Multiqudit Supersinglet Slater States via Constant Number of Binary Measurements}
\author[1]{Arturo Konderak\thanks{Email: \texttt{akonderak@cft.edu.pl}}}
\author[1]{Wojciech Bruzda}
\author[1]{Remigiusz Augusiak}
\affil[1]{Center for Theoretical Physics, Polish Academy of Sciences, Aleja Lotnik\'{o}w 32/46, 02-668 Warsaw, Poland}

\maketitle
\abstract{Self-testing is a powerful device-independent technique that enables one to deduce the forms of both the quantum state and the measurements involved in a physical experiment based solely on observed correlations. Although numerous schemes for self-testing multipartite entangled states have been proposed, they are typically difficult to implement experimentally, as their complexity increases significantly with the number of subsystems or the local dimension. In this work, we introduce the first self-testing scheme of a relevant class of multiqudit genuinely entangled states that exploits only a constant number of binary measurements per observer, which significantly reduces the experimental effort to implement the scheme. Specifically, it enables the self-testing of multipartite Slater (or supersinglet) states composed of $d$ qu\textit{d}its with odd $d$ using only four two-outcome measurements per observer. Moreover, we prove that our scheme is robust to noise and experimental imperfections. For systems of even local dimension $d$, we also provide an adapted version of the scheme that requires only $d$ binary measurements per observer.}

\section{Introduction}

Bell nonlocality—the existence of correlations that cannot be reproduced by means of local hidden variable models~\cite{BELL2001}—is one of the most distinctive features of quantum theory~\cite{RevModPhys.86.419}. Indeed, the pioneering experiments~\cite{Clauser,Zeilinger,Aspect} to verify the existence of this phenomenon were awarded a Nobel Prize in Physics in 2022. In addition, beyond its fundamental interest, Bell non-locality has been recognized as a valuable resource for various applications within device-independent quantum information processing~\cite{PhysRevLett.98.230501}, where observers process information without placing trust in the devices themselves. One of such applications that has recently gained substantial interest, driven by the recent rapid development of new quantum technologies such as quantum cryptography systems or quantum computing, is the device-independent certification of entangled quantum states and the measurements performed on them, commonly referred to as self-testing~\cite{Mayers2004} (see also Ref.~\cite{Supic2020} for a review). Self-testing allows one to almost completely characterize (up to certain well-known equivalences) the distributed quantum state and the measurements performed on its local subsystems solely from the observed non-local correlations. It thus requires making only minimal assumptions about the internal working of the devices under study, which makes DI certification schemes immune to hacking exploiting particular specifics of the devices. 

Notably, apart from its role in certifying entangled states, self-testing has also found other interesting applications. For instance, self-testing results are often used in device-independent (DI) certification of randomness generated from quantum states~\cite{PhysRevA.93.040102,PhysRevResearch.2.042028}, and also to prove separation between various sets of quantum correlations in certain Bell scenarios (see, e.g., Ref.~\cite{Beigi2021separationof} and references therein).

For these reasons self-testing has attracted a lot of attention in recent years and a plethora of DI certification schemes for pure quantum states (both bipartite and multipartite) and quantum measurements have been designed. In the bipartite case a general scheme allowing to self-test an arbitrary pure entangled two-qudit state was introduced in Ref.~\cite{Coladangelo2017}. However, this scheme is not universally optimal in terms of the number of measurements that the parties must perform or the number of outcomes, and more optimal schemes for the maximally entangled states have been proposed (see, e.g., Refs.~\cite{Sarkar2021,Mancinska2024}). At the same time, in the multipartite case very general methods to certify any arbitrary multiqubit states
have been proposed first in Ref.~\cite{Supic2023}, which is based on quantum networks, and then in Ref.~\cite{Balanzo-Juando2024} in the standard Bell scenario. 
Then, several other schemes have been proposed for various specific classes of multipartite states such as multiqubit graph states~\cite{Baccari2020scalable}, the $W$ state~\cite{Wu2014}, or the Schmidt states of arbitrary local dimension~\cite{Supic2018}.

However, the common problem permeating most of the existing schemes is that their complexity in terms of the number of measurements or outcomes grows with the system size (either the local dimension or the number of subsystems or both). On one hand, 
the schemes designed to certify states of arbitrary local dimension such as~\cite{Coladangelo2017,Supic2018,Sarkar2021} require performing measurements whose number of outcomes grows system's local dimension and, on the other hand, as far as multiqubit states are concerned, the method of Ref.~\cite{Balanzo-Juando2024} requires
a number of measurements that grows as factorial in the number of subsystems, whereas
the scheme of Ref.~\cite{Supic2023} requires doubling the number of parties in order to certify a state. At the same time, for practical applications of self-testing, the experimental effort required for their implementation needs to be reduced as much as possible. Only recently, a novel approach to DI certification was proposed in Ref.~\cite{Mancinska2024}, exploiting certain algebraic methods, which enables designing a robust self-testing scheme for maximally entangled states of arbitrary local dimension based on constant-sized statistics. In this scheme, both parties need to perform only four two-outcome measurements to certify the state, regardless of its local dimension. However, no such scheme exists in the multipartite scenario.

In this paper {we address the above problem and} prove that it is possible to {design a scheme allowing to} self-test {multipartite states of arbitrary local dimension and number of subsystems just from a constant number of binary measurements performed by each observer.} {More precisely, we show how to self-test the so-called Slater (determinant) states ~\cite{carlen2016entropy}, also known as supersinglet states~\cite{Cabello2003} of odd local dimension, by locally measuring only four two-outcome observables}. {These states are defined in the antisymmetric tensor product of $d$ $d$-dimensional Hilbert spaces and are the only zero-eigenstates of the respective total spin operator. Simultaneously, they have found numerous applications in quantum chemistry~\cite{szabo1996modern,Giuliani2025}, quantum computation~\cite{Lidar1998} or quantum information~\cite{PhysRevLett.89.100402,Piccolini_2025} and} also exhibit the singlet property~\cite{Bernards2024}, which protects them against decoherence effects~\cite{Amato2025}, provided that the underlying interaction preserve this symmetry~\cite{Lidar1998,Cabello2007}. {We also obtain a self-testing statement for the case of even $d$, however, in this case the observers require measuring $d$ binary observables. Finally, we demonstrate that self-testing is robust against noise, a fundamental property for the experimental implementation of such protocols~\cite{Gigena2024}.}

{In order to derive our results, we generalize the techniques of Ref.~\cite{Mancinska2024} to the multipartite setting. In particular, we provide an explicit and direct construction of the operators representing the measurement elements of the optimal measurements giving rise to the expected nonlocal correlations.
This improves the construction used in Ref.~\cite{Mancinska2024}, which is based on induction.}

{Let us finally notice that} prior to this work, a self-testing strategy for the same state was very recently introduced in~\cite{Saha2024}. This approach provides an optimal strategy—one that ensures a quantum game can be won with absolute certainty—a feature linked to Kochen-Specker sets in the context of local measurements~\cite{Cabello2025,Trandafir2024}. However, that method requires a number of local bipartite measurements that increase exponentially with the Hilbert space dimension, making it less efficient than the strategy proposed in this work.

The article is organized as follows. Section~\ref{sec:preliminaries} covers the preliminaries, with Subsection~\ref{sec:self_testing} introducing the notions of Bell non-locality and self-testing, along with a brief summary of the state of the art. Section~\ref{sec:math_back} presentes the mathematical tools used throughout the paper. In particular, Subsection~\ref{sec:algebras} provides a concise guide to $C^*$-algebras, while Subsection~\ref{sec:identity_sum} discusses the properties of projections summing to a scalar times the identity. Explicit constructions of the latter families are provided in Section~\ref{sec:explicit_construction}. In Section~\ref{sec:self_testing_Slater}, we present the strategy for self-testing the Slater state. In Section~\ref{sec:self_testing_measurement}, we show that the measurements are self-tested in the general case, while in Section~\ref{sec:self_testing_state}, we describe sufficient conditions for self-testing, and specifically two such conditions, namely Conjectures~\ref{prop:conjecture} and~\ref{prop:conjecture_3}. In Subsection~\ref{Section:Algebraic_approach} we prove that the conjectures are satisfied in two specific conditions, which allow us to obtain the self-testing statement in two configurations. Finally, in Section~\ref{sec:self-testing_robust}, we prove that the self-testing statements are all robust.

\section{Preliminaries}\label{sec:preliminaries}

In this Section, we provide a comprehensive overview of the key definitions and properties underlying Bell non-locality and self-testing. We then introduce the mathematical definition of Slater states and discuss their symmetry properties.

\subsection{Bell scenario}\label{sec:self_testing}
Bell non-locality provides a certification technique to ensure that two or more spatially separated parties share a quantum resource~\cite{Augusiak2014}. Let $n$ parties share a quantum state $\ket{\Psi}$, belonging to the composite Hilbert space $\mathcal{H}:=\mathcal H_0\otimes\mathcal{H}_1\otimes\ldots\otimes\mathcal{H}_{n-1}$, where $\mathcal{H}_k$ is a finite-dimensional Hilbert space associated to the $k$th party. Each party $k=0,\dots, n-1$ has access to a measurement device performing one of $N$ possible measurements, each described by a two-element POVM: $\{E_{\mu,0}^{(k)},E_{\mu,1}^{(k)}\}$, where $\mu=0,\dots,N-1$. The measurement elements $E_{\mu,0}^{(k)}$ and $E_{\mu,1}^{(k)}$ are positive operators in $\mathcal B(\mathcal H_k)$ that sum up to the identity $\mathbb I_{\mathcal H_k}$ on $\mathcal H_k$. For simplicity, we consider measurements with only two outcomes (corresponding to indices $0$ and $1$ in the measurement operators), although the formalism extends naturally to more outcomes.

Based on a family of inputs $(\mu_0,\mu_1,\dots,\mu_{n-1})\in\{0,1,\dots,N-1\}^{n}$, each party $k$ performs the measurement $\{E_{\mu_k,0}^{(k)},E_{\mu_k,1}^{(k)}\}$, thus obtaining a total probability distribution
\begin{equation}\label{eq:correlation_general}
    {p}(a_0,a_1,\dots,a_{n-1}|\mu_0,\mu_1,\dots,\mu_{n-1})=\mel{\Psi}{E_{\mu_0,a_0}^{(0)}\otimes\dots\otimes E_{\mu_{d-1},a_{d-1}}^{(d-1)}}{\Psi}
\end{equation}
with $(a_0,a_1,\dots,a_{n-1})\in\{0,1\}^{n}$. In order to ease the notation, we write $\boldsymbol a=(a_0,\dots,a_{n-1})$, $\boldsymbol \mu=(\mu_0,\dots,\mu_{n-1})$ and $E_{\boldsymbol \mu,\boldsymbol a}={E_{\mu_0,a_0}^{(0)}\otimes\dots\otimes E_{\mu_{d-1},a_{d-1}}^{(d-1)}}$. We also denote $[a,b]=\{a,a+1,\dots,b-1\}$ for any pair $a,b\in\mathbb N$ such that $a<b$; in the particular case of $a=0$ it reduces to $[b]=[0,b]$. The collection of probability distributions $\{p(\boldsymbol a|\boldsymbol\mu)\}$ is called a quantum correlation generated by the state $\ket{\Psi}$ and the measurements $\{E_{\boldsymbol \mu,\boldsymbol a}\}$.

The correlation $p(\boldsymbol a |\boldsymbol \mu)$ is said to be Bell non-local if it cannot be reproduced using only local hidden variable models (LHV). Explicitly, it cannot be written in the form
\begin{equation}\label{eq:local_corr}
   p_{\mathrm{loc}}(\boldsymbol{a}|\boldsymbol{\mu})= \int_{\lambda\in\Omega} d\lambda \ \omega(\lambda) \ p_0(a_0|\mu_0,\lambda)\cdots p_{N-1}(a_{N-1}|\mu_{N-1},\lambda),
\end{equation}
with $p_k(a_k|\mu_k,\lambda)$ representing a local strategy, and $\lambda$ a shared randomness distributed according to a probability density $\omega(\lambda)$. Correlations admitting LHV models as in~\eqref{eq:local_corr} are also called local or classical. Thus, the existence of {non-local} correlations between spatially separated parties serves as proof of the quantum nature of a system. 

\subsection{Self-testing}\label{sec:self-testing_intro}
There are scenarios in which Bell non-locality not only certifies the quantum nature of correlations but also fully (up to certain equivalences) characterizes the underlying entangled quantum state and measurements. This phenomenon is known as self-testing~\cite{Mayers2004} (see also Ref.~\cite{Supic2020} for a review). To formalize this concept, suppose a certain reference state $\ket*{\tilde \Psi}\in\tilde{\mathcal{H}}_0\otimes\ldots\otimes\tilde{\mathcal{H}}_{n-1}$ and a reference family of measurements $\{\tilde E_{\boldsymbol{\mu},\boldsymbol{a}}\}$ generate the correlations $\tilde p(\boldsymbol{a}|\boldsymbol{\mu})$ as in Equation~\eqref{eq:correlation_general}. Suppose these are the same correlations generated by a Bell experiment involving an unknown quantum state $\ket{\Psi}\in\mathcal{H}_0\otimes\ldots\otimes\mathcal{H}_{n-1}$
and unknown quantum measurements $\{E_{\boldsymbol \mu,\boldsymbol a}\}$. Then, the definition of self-testing reads as follows.
%If this implies that $\ket{\Psi'}$ and $\{E'_{\boldsymbol %\mu,\boldsymbol a}\}$ coincide with $\ket{\Psi}$ and $\%
%{E_{\boldsymbol \mu,\boldsymbol a}\}$, we say that the %correlation $p(\boldsymbol a|\boldsymbol{\mu})$ self-tests %the apparatus. Formally:
%
% \deleted{We say that the quantum correlation $p(\boldsymbol a |\boldsymbol \mu)$ is further self-testing the state and the measurements if it can be generated by a unique quantum strategy, up to local isometries. We clarify this notion in the following definition.}
%
\begin{definition}[Self-testing]
We say that the observed correlations $p(\boldsymbol a|\boldsymbol \mu)$ obtained from a quantum state $\ket{\Psi}$ and measurements $\{E_{\boldsymbol \mu,\boldsymbol a}\}$ self-test the reference quantum realization $\ket*{\tilde\Psi}$ and $\{\tilde E_{\boldsymbol \mu,\boldsymbol a}\}$ if for all $k\in[n]$ there exists some auxiliary Hilbert spaces $\mathcal I_k$ and isometries $V_k:\mathcal H_{k}\rightarrow \tilde{\mathcal H}_k\otimes \mathcal I_{k}$ such that
\begin{align}
        \label{eq:self-testing_state}
        \left(\bigotimes_{k=0}^{n-1} V_k \right)\ket{\Psi}=\ket*{\tilde\Psi}\otimes \ket{\psi_{\mathrm{junk}}},
    \end{align}
    and for all $\boldsymbol a\in\{0,1\}^n$ and $\boldsymbol \mu\in[N]^{n}$
    \begin{equation}\label{eq:self-testing_state_measurements_2}
       \left(\bigotimes_{k=0}^{n-1} V_k \right) E_{\boldsymbol \mu,\boldsymbol a} \ket{\Psi}= \tilde E_{\boldsymbol \mu,\boldsymbol a}\ket*{\tilde\Psi}\otimes \ket{\psi_{\mathrm{junk}}},
    \end{equation}
where $\ket{\psi_{\mathrm{junk}}}$ is some auxiliary state from $\mathcal I_{0}\otimes\ldots\otimes \mathcal I_{n-1}$.

\end{definition}
We focus exclusively on cases where all states are pure, since—thanks to dilation—this assumption can always be made without any loss of generality~\cite{Baptista2023}.

We further say that the self-testing is \textit{robust} if the relations~\eqref{eq:self-testing_state} and~\eqref{eq:self-testing_state_measurements_2} hold approximately for correlations close to $\tilde p(\boldsymbol a|\boldsymbol \mu)$. Specifically, for every  $\varepsilon>0$ there exists $\delta$ such that for all quantum states $\ket{\Psi}$ and measurements $\{E_{\boldsymbol{\mu}|\boldsymbol{i}}\}$ producing correlations $p(\boldsymbol a|\boldsymbol \mu)$ such that $\norm{p-\tilde p}_1:=\sum_{\boldsymbol a\boldsymbol \mu}|p(\boldsymbol a|\boldsymbol \mu)-\tilde p(\boldsymbol a|\boldsymbol \mu)|<\delta$, conditions~\eqref{eq:self-testing_state} and~\eqref{eq:self-testing_state_measurements_2} hold up to an error $\varepsilon$ in the one-norm. Robust self-testing is essential because it ensures that theoretical certifications of quantum systems remain valid under realistic, noisy experimental conditions. 

% Lastly, let us call $W^{(k)}$, $W^{'(k)}$ and $W_{\mathrm{junk}}^{(k)}$ the projections onto the support of the reduced density matrices of $\ket{\Psi'}$, $\ket{\Psi}$ and $\ket{\psi_{\mathrm{junk}}}$ respectively. Observe that equations~\eqref{eq:self-testing_state} and~\eqref{eq:self-testing_state_measurements_2} imply that 
% \begin{equation}\label{eq:self-testing_measurement}
%      V_k\tilde E_{\mu,a}^{'(k)}V_k^*= \tilde E_{\mu,a}^{(k)}\otimes W^{(k)}_{\mathrm{junk}},\quad k\in [n], a\in\{0,1\},
% \end{equation}
% with $\tilde E_{\mu,a}^{'(k)}=W^{'(k)}E_{\mu,a}^{'(k)}W^{'(k)}$ and $\tilde E_{\mu,a}^{(k)}=W^{(k)}E_{\mu,a}^{(k)}W^{(k)}$.

Assuming, without loss of generality, that $\ket{\Psi}$, $\ket*{\tilde \Psi}$ and $\ket{\psi_{\mathrm{junk}}}$ have maximal support over the reduced Hilbert spaces $\mathcal H_k$, $\tilde{\mathcal H}_k$ and $\mathcal I_k$, Equations~\eqref{eq:self-testing_state} and~\eqref{eq:self-testing_state_measurements_2} imply that~\cite{Baptista2023}:
\begin{equation}\label{eq:self-testing_measurement}
     V_k E_{\mu,a}^{(k)}V_k^*= \tilde E_{\mu,a}^{(k)}\otimes \mathbb I_{\mathcal I_k}\qquad \text{for}\qquad k\in[n], \quad \mu\in [N], \quad a\in\{0,1\}.
\end{equation}
In particular, if only Equation~\eqref{eq:self-testing_measurement} holds, we say that the quantum correlation $p(\boldsymbol{a}|\boldsymbol{\mu})$ self-tests the measurement.

\begin{remark}
It is worth noting that there is one more degree of freedom/equivalence that should be considered in the definition of self-testing. Namely, the state and measurements can only be self-tested up to complex conjugation~\cite{Supic2020}. However, in our case, both the state and measurements are real, and this equivalence need not be accounted for.
\end{remark}

In the bipartite case ($n=2$), the maximally entangled two-qubit states can be self-tested via maximal CHSH inequality violation~\cite{Clauser1969}, while all pure two-qubit entangled states are self-testable using tilted CHSH inequalities~\cite{Yang2013, Bamps2015}. More generally, any pure bipartite state can be self-tested using three or four measurements per party, each with $d$ outcomes (where $d$ is the state's Schmidt number), scaling linearly with Hilbert space dimension~\cite{Coladangelo2017}.

The general self-testing statement in~\cite{Coladangelo2017} describes a procedure of self-testing any pure entangled state shared by two parties. Nevertheless, for certain states, more effective strategies can be developed. The SATWAP inequality~\cite{Salavrakos2017} allows for self-testing by maximal violation the maximally entangled state~\cite{Sarkar2021} between two parties in any dimensions with only two measurements of $d$ outcomes. More recently, it has been found in~\cite{Mancinska2024} that the maximally entangled state can be self-tested with a bounded {and independent of the local dimension of the system} number of {four} measurements and {two} outcomes {per observer}. To achieve this, techniques developed in~\cite{Rabanovich2000,Kruglyak2002,Kruglyak2003} were considered, specifically the characterization of projections based on their sum.

% \textbf{(Remik: The following sentence is not true. There is a paper by Supic et al. This paper should be rewritten so that all the main results are mentioned including that of Supic, however, schemes for multi-qudit systems are missing which are implementable with constant-sized statistics are missing.)} \rev{Are you referring to this one\cite{Supic2023}? I didn't see the details, but it appears to me that they introduce an alternative notion of ST, based on network.}
{Some progress has been made in developing self-testing statements for the multipartite quantum systems}. The current, more general self-testing statement pertains to qubit states~\cite{Balanzo-Juando2024}, and it demonstrates that, up to a phase conjugation, all pure qubit states can be self-tested. Other specific self-testing statements are given here for qubits. In~\cite{Wu2014}, a self-testing statement is presented for the tripartite $W$ state, which is performed using two measurements and two outcomes per party. In~\cite{Supic2018}, the GHZ state is self-tested with a total of two measurements and two outcomes in any number of parties. In~\cite{Fadel2017}, self-testing of all $N$ parties Dicke states has been presented for many copies, and self-testing is achieved with $2$ local measurements. In~\cite{Pal2014}, three local measurements with two outcomes are used to self-test the three-qubit $W$ state and the three and four-qubit GHZ states. In~\cite{Baccari2020}, a self-testing for graph states is presented, consisting of 2 local measurements with two outcomes. Finally, Ref.~\cite{Supic2018} presents a self-testing statement for all Schmidt qudit states, with $d$ measurements per party. Recently, a procedure for self-testing the Slater state was introduced~\cite{Saha2024}, where the number of measurements is related to the number of rigid Kochen-Specker sets. In particular, the number of measurements can be found to grow exponentially with the number of parties. On the contrary, the procedure presented in this work offers a significant improvement by allowing self-testing with only $4$ dichotomic measurements, regardless of the number of parties.

\subsection{Slater states}

Assume that each local Hilbert space $\mathcal{H}_k$ is $d$-dimensional, that is, $\mathcal{H}_k\equiv\mathcal{H}=\mathbb{C}^d$. The set of bounded operators over $\mathcal H$ is denoted by $\mathcal B(\mathcal H)$, and it can be identified with $\mathcal M_d$, set of $d\times d$ complex matrices.
% \rev{consider uniforming this part to the general case of the invariant subspaces under product.}
Consider $n=d$ copies of {the local $d$-dimensional system}, meaning that the total Hilbert space is $\mathcal{H}_{\mathrm{tot}}=\mathcal H^{\otimes d}$. Denoting by $\{\ket{j}\}_{j=0,\dots,d-1}$ the standard (computational) basis in $\mathcal H$, the Slater state is defined as
\begin{equation}\label{eq:slater_state_canonical_basis}
    \ket{\Psi_{\mathrm{S}}}=\frac{1}{\sqrt{d!}}\sum_{\sigma \in\mathbb P_d}\mathrm{sign}(\sigma)\ket{\sigma(0)\otimes\sigma(1)\otimes\dots\otimes\sigma(d-1)}\in\mathcal H^{\otimes d},
\end{equation}
where $\mathbb P_n$ is the set of permutations of order $n$, and $\mathrm{sign}(\sigma)$ is the sign of the permutation. The Slater state satisfies the following singlet property
\begin{equation}\label{eq:singlet_property}
    U^{\otimes n}\ket{\Psi_{\mathrm S}}=\xi\ket{\Psi_{\mathrm S}},\quad \forall\ U\in\mathcal{U},
\end{equation}
with $\xi$ being a phase, and $\mathcal U(d)$ the set of unitaries over $\mathbb C^d$. The Slater state is the only vector in $\mathcal H^{\otimes d}$ satisfying this property~\cite{Bernards2024}, as we discuss in Sec.~\ref{sec:self_testing_measurement}.

The Slater states are the least entangled fermionic states when considering a partition with only one party~\cite{Coleman1963}, while this is not generally true for the reduced density matrix over two parties~\cite{Yang1962,carlen2016entropy}.

Due to~\eqref{eq:singlet_property}, a change of basis results in the Slater state acquiring an overall phase factor. Indeed, this turns out to be the only state in the antisymmetric subspace of $\mathcal H^{\otimes d}$.

\section{Mathematical background}\label{sec:math_back}
In this Section, we bring together several definitions and properties which are fundamental for the self-testing statement. First, we present a couple of properties regarding finite-dimensional $C^*$-algebra, and discuss the notion of projections adding up to the identity times a constant.

\subsection{\texorpdfstring{$C^*$-algebras}{}}\label{sec:algebras}
% \rev{I am considering giving the general formal definition of a $C^*$-algebra.}

The complete abstract definition of $C^*$-algebra is rather intricate~\cite{Bratteli1987,Blackadar2006}, hence we shall reduce the considerations to the case of $C^*$-algebras of operators over a separable Hilbert space $\mathcal H$. Let $\mathfrak A$ be a set of operators over a Hilbert space $\mathcal H$, namely, $\mathfrak A\subset \mathcal B(\mathcal H)$. Then, $\mathfrak A$ is called a $C^*$-subalgebra of $\mathcal B(\mathcal H)$, if it is a complex vector space, closed under composition of Hermitian adjoint of an element. Furthermore, it must be closed in the uniform norm $\norm{\cdot}_{\infty}$, defined as
\begin{equation}
    \norm{A}_{\infty}=\sup\{\norm{A\ket{\psi}}:\ket{\psi}\in\mathcal H, \norm{\ket{\psi}}=1\}.
\end{equation}

% Two $C^*$-algebras $\mathfrak A$ and $\mathfrak B$ are said isomorphic if there exists a bijection $\phi:\mathfrak A\rightarrow \mathfrak B$ preserving products, linear combinations, and adjoint operations.

In particular, a structure theorem holds for $C^*$-algebra~\cite[p.~75]{Davidson1996}. 
\begin{theorem}[Structure Theorem]
Let $\mathfrak A$ be a $C^*$-algebra over a finite dimensional Hilbert space $\mathcal H$. Then, it is possible to decompose $\mathcal H$ as
\begin{equation}\label{eq:dec_hilbert_space_algebra}
    \mathcal H=\mathcal H_0 \oplus \bigoplus_{k=1}^{M} \mathcal H_{k,1}\otimes \mathcal H_{k,2},
\end{equation}
with $\mathcal H_0$, $\mathcal H_{k,1}$ and $\mathcal H_{k,2}$ finite dimensional Hilbert spaces, and $k=1,2,\dots, M$. With respect to decomposition~\eqref{eq:dec_hilbert_space_algebra}, the $C^*$-algebra $\mathfrak A$ can be expressed in the form
\begin{equation}\label{eq:structure_theorem_c_algebras}
    \mathfrak A = 0\oplus \bigoplus_{k=1}^M \mathcal B(\mathcal H_{k,1}) \otimes \mathbb I_{k,2},
\end{equation}
where, for each $k=1,\dots,M$, $\mathbb I_{k,2}$ represents the identity over the Hilbert space $\mathcal H_{k,2}$, while $0$ is the null operator over $\mathcal H_0$.
\end{theorem}
If $\mathfrak A$ represents the set of observables of a theory, the direct sum implies the existence of superselection rules, the sets of operators $\mathcal B(\mathcal H_{k,1})$ contains the coherent superposition, and finally the identities $\mathbb I_{k,2}$ involves gauge degrees of freedom~\cite{Balachandran2019}.
Given a set of operators $S\subset \mathcal B(\mathcal H)$, we can define its commutant $S'$
\begin{equation}
    S'=\{X\in\mathcal B(\mathcal H):XY=YX,\ \forall Y\in S\}.
\end{equation}
The bicommutant $S''$ of $S$ is defined as the commutant of $S'$. For a $C^*$-algebra $\mathfrak A$ in the form of~\eqref{eq:structure_theorem_c_algebras}, the commutant and bicommutant are given by
\begin{align}
    \mathfrak A' &= \mathcal B(\mathcal H_0) \oplus \bigoplus_{k=1}^M \mathbb I_{k,1} \otimes \mathcal B(\mathcal H_{k,2}),\\
    \mathfrak A'' &=\mathbb I_0 \oplus \bigoplus_{k=1}^M \mathcal B(\mathcal H_{k,1}) \otimes \mathbb I_{k,2},
\end{align}
with $\mathbb I_0$ and $\mathbb I_{k,1}$ the identities over $\mathcal H_0$ and $\mathcal H_{k,1}$ respectively, for $k=1,\dots, M$.

In particular, for $\mathcal H_{0}=0$, we have that $\mathfrak A''=\mathfrak A$, which is the finite-dimensional version of the von Neumann bicommutant theorem~\cite{Bratteli1987,V_Neumann1930}.

\subsection{Projections summing to the identity}\label{sec:identity_sum}
In this Section, we summarize a specific class of projective operators which are particularly relevant for several self-testing problems~\cite{Dykema2019,Mancinska2024,Volcic2024}. This class of was studied in~\cite{Rabanovich2000,Kruglyak2002,Kruglyak2003}, and here we recall the construction and the most important properties.

Following~\cite{Kruglyak2002}, the main motivation is to find classes of orthogonal projections summing up to the identity. Specifically, let $N$ be an integer, and the problem boils down to identify all scalars $x\in\mathbb{R}$ for which there are $N$ orthogonal projections $P_0$, $P_1,\dots,P_{N-1}$ on a separable Hilbert space $\mathcal H$ summing up to $x \mathbb I_{\mathcal H}$.
%We summarize the main results in the following theorem, see Theorem~5.1 and Lemma~5.3 from~\cite{Mancinska2024}.

\begin{theorem} [Theorem~5.1 and Lemma~5.3 from~\cite{Mancinska2024}]\label{th:mancinska}
    Let $\Lambda_3=\{3/2\}$ and, define, for $N\geqslant 4$, the set $\Lambda_N=\{x_k\}_{k=0}^\infty$ recursively as
    \begin{equation}\label{eq:set_lambda_N}
        x_0=0,\quad x_k = 1+\frac{1}{N-1-x_{k-1}},\quad k\geqslant 1.
    \end{equation}
    Given $x\in\Lambda_N$, write it as $x=b/d$, with $b$ and $d$ natural numbers in the lowest form. Then, we can find $N$ orthogonal projections $\{P_0,\dots,P_{N-1}\}$ in $\mathcal M_d$ such that
    \begin{equation}\label{eq:operators_sum_identity}
        \sum_{\mu=0}^{N-1}P_{\mu}=x\mathbb I_d,
    \end{equation}
    with $\mathbb I_d$ being the identity in $\mathcal M_d$. Furthermore, these projections are uniquely defined, up to an isometry. Specifically, suppose that there are $N$ orthogonal projections $\{{E}_1,\dots,E_N\}$ over an Hilbert space $\mathcal H$ satisfying~\eqref{eq:operators_sum_identity}, namely
    \begin{equation}\label{eq:operators_sum_identity_2}
        \sum_{\mu=0}^{N-1}  E_\mu=x\mathbb I_\mathcal H,
    \end{equation}
    with $\mathbb I_{\mathcal H}$ being the identity over $\mathcal H$. Then, up to a unitary transformation, $\mathcal H=\mathbb C^d\otimes \mathcal H'$, for some Hilbert space $\mathcal H'$, and for all $\mu=0,\dots, N-1$ we have $P_\mu=P'_\mu \otimes \mathbb I_{\mathcal H'}$.
\end{theorem}

Roughly speaking, the idea behind the self-testing statements is based on the fact that we can infer, from the quantum correlations, both that the operators $P_j$ are orthogonal projections, and that they sum to a scalar multiple of the identity. In this sense, employing synchronous measurement is essential. The explicit form of the projections $\{P_0,\dots,P_{N-1}\}$ can be obtained by induction in the general case, see~\cite[Appendix B]{Mancinska2024}, and it is based on the idea of constructing functors between different representations of a specific $C^*$-algebras~\cite{Kruglyak2002}. It is possible to prove that the $C^*$-algebra generated by $\{P_0,\dots,P_{N-1}\}$ is the whole algebra of square matrices $\mathcal M_d$.
In the following Section, we present two explicit constructions for the projections summing up to the identity, which enable explicit self-testing of the Slater state in specific scenarios. 

\section{Explicit construction of projections summing up to the identity}\label{sec:explicit_construction}

In this Section, we present the construction of projections that sum up to the identity as in Theorem~\ref{th:mancinska}, in two situations. In the first case, we consider the construction corresponding the projections $\{P_0,\dots,P_{N-1}\}$ having rank $1$. In the second case, we provide the construction for the case $N=4$, which allows for self-testing with a bounded number of measurements. This construction, inspired by~\cite[Conj.~5.10]{Wu1994},
is more efficient than the one presented in Ref.~\cite{Mancinska2024}, as it does not use induction.

\subsection{\texorpdfstring{Case $N=d+1$}{}}\label{sec:rank_1_example}

    The first relevant example corresponds to the case of the projections $\{P_\mu\}_{\mu=0}^{N-1}$ having rank $1$. They exist in all dimensions $d$ by choosing $N=d+1$, and $x=N/(N-1)$ in Equation~\eqref{eq:set_lambda_N}, with $N\geqslant 4$. Interestingly, this class of projections is employed to construct all the other representations using the functor defined in~\cite{Kruglyak2002}.

    \begin{proposition}
        Let $\mathcal H=\mathbb C^d$, and let $\{\ket{\mu}\}_{\mu=0,\dots,d-1}$ be an orthonormal basis of $\mathcal H$. Define the states
        \begin{equation}\label{eq:decomposition_induction}
            \ket{\psi_\mu}= \sum_{\nu=0}^{\mu-1} \alpha_\nu \ket{\nu}+\beta_\mu\ket{\mu},\quad \mu=0,\dots,d-1,
        \end{equation}
        and
        \begin{equation}\label{eq:decomposition_induction_2}
            \ket{\psi_d}=\sum_{\nu=0}^{d-1} \alpha_\nu \ket{\nu},
        \end{equation}
        with
        \begin{align}\label{eq:decomposition_induction_3}
            \alpha_\mu =-\sqrt{\frac{d+1}{d(d-\mu)(d-\mu+1)}},\qquad \beta_\mu =+\sqrt{\frac{(d+1)(d-\mu)}{d(d-\mu+1)}}.
        \end{align}
        Then, the rank $1$ projections $P_\mu=\ketbra{\psi_\mu}$, with $\mu=0,\dots, d$, satisfy
        \begin{equation}\label{eq:sum_operators_rank_1}
            \sum_{\mu=0}^{d} P_\mu=\frac{d+1}{d}\mathbb I_{d},
        \end{equation}
        and therefore correspond to the $d+1$ orthogonal projections summing up to $(d+1)/d$.
    \end{proposition}
    \begin{remark}
            The form of the vectors~\eqref{eq:decomposition_induction} can be inferred by observing that, from Equation~\eqref{eq:operators_sum_identity} and the uniqueness condition, for $i\neq j$ we have that $\tr(P_\mu P_\nu)=1/d^2$, which implies that $|\braket{\psi_\mu}{\psi_\nu}|=1/d$. By using the Gram-Schmidt procedure on $\{\ket{\psi_\mu}\}_{\mu=0,\dots,d-1}$, and requiring the scalar product to be $-1/d$, we obtain Equations~\eqref{eq:decomposition_induction},~\eqref{eq:decomposition_induction_2} and~\eqref{eq:decomposition_induction_3}.
    \end{remark}
    \begin{proof}
        First, observe that the states $\ket{\psi_\mu}$ are correctly normalized as 
    \begin{equation}
        \sum_{\nu=0}^{\mu-1}\alpha_\nu^2=\frac{\mu}{d(d-\mu+1)}=1-\beta_\mu^2,\quad \mu=1,2,\dots,d.
    \end{equation}
    The latter can be proved by induction. For $\mu=1$ it is trivial. Suppose it holds for $\mu-1$; we then have
    \begin{align}
        \sum_{\nu=0}^\mu \alpha_\nu^2=\frac{\mu}{d(d-\mu+1)}+\frac{d+1}{d(d-\mu)(d-\mu+1)}=\frac{\mu+1}{d(d-\mu)}.
    \end{align}
    Choose $\mu<\mu'$, and we can easily show that:
    \begin{align}
        \braket{\psi_\mu}{\psi_{\mu'}}=\sum_{\nu=0}^{\mu-1} \alpha_\nu^2+\alpha_\mu\beta_\mu=\frac{\mu}{d(d-\mu+1)}-\frac{d+1}{d(d-\mu+1)}=-\frac{1}{d}.
    \end{align}
    Finally, for all $\mu,\mu'$, we have
    \begin{equation}\label{eq:value_identity}
        \mel{\psi_\mu}{\left(\sum_{\nu=0}^{d}\ketbra{\psi_\nu}\right)}{\psi_{\mu'}}=
        \begin{cases}-\frac{d+1}{d^2},\quad & \mu\neq \mu'\\ \frac{d+1}{d},\quad & \mu=\mu'
        \end{cases} = \frac{d+1}{d}\braket{\psi_\mu}{\psi_{\mu'}}.
    \end{equation}
    From Equation~\eqref{eq:decomposition_induction} we know that $\{\ket{\psi_\mu}\}_{\mu=0,\dots,d-1}$ span all $\mathbb C^d$, and consequently, Equation~\eqref{eq:value_identity} implies~\eqref{eq:sum_operators_rank_1}.
    \end{proof}

\subsection{\texorpdfstring{Case $N=4$}{}}\label{sec:4_proj}
We now present the construction of the four projections summing up to the identity. Take $N=4$ in Equation~\eqref{eq:set_lambda_N}, and observe that the elements in $\Lambda_4$ can be expressed as
\begin{equation}
    x_k=\frac{4k}{2k+1}, \qquad\forall k\geqslant 0.
\end{equation}
From Theorem~\ref{th:mancinska}, for all $k$ it is possible to find $4$ projections in dimension $d=2k+1$ satisfying~\eqref{eq:operators_sum_identity}, with $x=x_k$. 

Let $\mathcal H=\mathbb C^d$ be a Hilbert space of dimension $d=2k+1$. We look for four orthogonal projections $\{P_0,P_1,P_2,P_3\}$ satisfying~\eqref{eq:operators_sum_identity}. It is convenient to consider the orthogonal complements $A,B,C,D$ of $P_0,P_1,P_2,P_3$ respectively, that is
\begin{align}
    A=\mathbb I_d - P_0,\quad B=\mathbb I_d - P_1,\quad C=\mathbb I_d - P_2,\quad D=\mathbb I_d - P_3.
\end{align}
These are orthogonal projections satisfying
\begin{equation}
    A+B+C+D=\frac{4k+4}{2k+1}\mathbb I_{2k+1}.
\end{equation}
One can readily infer the structure of all $P_\mu$ from $A,B,C,D$. In particular $A$, $B$, $C$ and $D$ can be expressed in a block-diagonal form, as a direct sum of $2\times 2$ projections, and they can be arranged in such a way that $A$, $B$ and $C$, $D$ share the same block structure. We precede the construction by presenting the following lemma.

\begin{lemma}\label{lemm:sum_of_projections}
    Let $X$ be a $2\times 2$ diagonal matrix of the form
    \begin{equation}
        X=\begin{pmatrix}
            1+x&0\\0&1-x
        \end{pmatrix}
    \end{equation}
    with $-1<x<1$. Then, $X$ can be written as the sum of two orthogonal projections $X=A+B$, $A=A^2=A^*$ and $B=B^2=B^*$, as
    \begin{equation}
        A=\frac{1}{2}\begin{pmatrix}
            {x_1} & {\sqrt{x_1x_2}}\\ {\sqrt{x_1x_2}}&{x_2}
        \end{pmatrix},\quad
        B=\frac{1}{2}\begin{pmatrix}
            {x_1} & -{\sqrt{x_1x_2}}\\ -{\sqrt{x_1x_2}}&{x_2} 
        \end{pmatrix},
    \end{equation}
    and $x_1=1+x$, $x_2=1-x$.
\end{lemma}
\begin{proof}
    Straightforward.
\end{proof}

To proceed with the construction of $A,B,C$ and $D$, let $\{\ket{\ell}\}_{\ell=0}^{2k}$ be the standard basis of $\mathbb C^{2k+1}$. Furthermore, define $E_{\ell,\ell'}=\ketbra{\ell}{\ell'}$, with $\ell,\ell'=0,\dots, 2k$, and $E_{\ell}=E_{\ell,\ell}$. The following decomposition holds:
\begin{align}
     \frac{4 k+4}{2k+1}\mathbb I_{2k+1}&=\sum_{\ell=0}^{2k}\frac{4 k+4}{2k+1}E_{\ell}=\sum_{\ell=0}^{k}\frac{4 k+4}{2k+1}E_{2\ell}+\sum_{\ell=0}^{k-1}\frac{4 k+4}{2k+1}E_{2\ell+1}\nonumber\\ &=\sum_{\ell=0}^{k}\left(\frac{4 \ell+2}{2k+1}E_{2\ell}+\frac{4k+2- 4\ell}{2k+1}E_{2\ell}\right)+ \sum_{\ell=0}^{k-1}\left(\frac{4k -4\ell}{2k+1}E_{2\ell+1}+\frac{4 \ell+4}{2k+1}E_{2\ell+1}\right)\nonumber\\
     &=\left(\sum_{\ell=0}^{k}\frac{4 \ell+2}{2k+1}E_{2\ell}+\sum_{\ell=0}^{k-1}\frac{4k -4\ell}{2k+1}E_{2\ell+1}\right)+\left(\sum_{\ell=0}^{k}\frac{4k+2- 4\ell}{2k+1}E_{2\ell}+\sum_{\ell=0}^{k-1}\frac{4 \ell+4}{2k+1}E_{2\ell+1}\right)\nonumber \\
     &=\sum_{\ell=0}^{k-1}\left(\frac{4 \ell+2}{2k+1}E_{2\ell}+\frac{4k -4\ell}{2k+1}E_{2\ell+1}\right)+\sum_{\ell=1}^{k}\left(\frac{4 \ell}{2k+1}E_{2\ell-1}+\frac{4k+2- 4\ell}{2k+1}E_{2\ell}\right)\nonumber\\&\quad+2E_0+2E_{2\ell}=\sum_{\ell=0}^{k-1}X_{\ell}+\sum_{\ell=1}^{k}Y_\ell + 2E_0+2E_{2k}.
\end{align}

Now, one observes that in each subspace $\ket{2\ell}\oplus \ket{2\ell+1}$, the operator $X_\ell$ has the form described in Lemma~\ref{lemm:sum_of_projections}, while in the subspace $\ket{2\ell-1}\oplus \ket{2\ell}$, the operator $Y_\ell$ does.

Define
\begin{align}
    x^{(1)}_\ell &= \frac{4 \ell+2}{2k+1},\quad x^{(2)}_\ell = \frac{4k- 4\ell}{2k+1},\\
    y^{(1)}_\ell &= \frac{4 \ell}{2k+1},\quad y^{(2)}_\ell = \frac{4k+2- 4\ell}{2k+1},
\end{align}
so that we can write
\begin{equation}
    X_\ell = A_\ell + B_\ell
\end{equation}
with
\begin{equation}\label{eq:operators_A_B}
    A_\ell =\frac{1}{2} \begin{pmatrix}x_\ell^{(1)}&\sqrt{x^{(1)}_\ell x_\ell^{(2)}}\\ \sqrt{x_\ell^{(1)}x_\ell^{(2)}} &x_\ell^{(2)}\end{pmatrix},\quad B_\ell=\frac{1}{2} \begin{pmatrix}x_\ell^{(1)}&-\sqrt{x_\ell^{(1)}x_\ell^{(2)}}\\ -\sqrt{x_\ell^{(1)}x_\ell^{(2)}} &x_\ell^{(2)}\end{pmatrix}.
\end{equation}
Similarly, we can write
\begin{equation}
    Y_\ell=C_\ell+D_\ell
\end{equation}
with
\begin{equation}\label{eq:operators_C_D}
    C_\ell =\frac{1}{2} \begin{pmatrix}y_\ell^{(1)}&\sqrt{y^{(1)}_\ell y_\ell^{(2)}}\\ \sqrt{y_\ell^{(1)}y_\ell^{(2)}} &y^{(2)}\end{pmatrix},\quad D_\ell=\frac{1}{2} \begin{pmatrix}y_\ell^{(1)}&-\sqrt{y_\ell^{(1)}y_\ell^{(2)}}\\ -\sqrt{y_\ell^{(1)}y_\ell^{(2)}} &y_\ell^{(2)}\end{pmatrix}.
\end{equation}

The previous computation shows that $A$ and $B$ can be written as a direct sum of matrices with support over $\ket{2\ell}\oplus \ket{2\ell+1}$, while $C$ and $D$ as the direct sum in the subspaces $\ket{2\ell-1}\oplus \ket{2\ell}$. Explicitly
\begin{equation}
\fboxsep=6pt
A=\begin{pmatrix}
{\fbox{\ensuremath{A_0}}} &  & &  & \\
 & {\fbox{\ensuremath{A_1}}} &  & & \\
  &  & \ddots & & \\
  &  & & \setlength{\fboxsep}{5pt}\fbox{\ensuremath{A_{k-1}}} & \\
  &  & &  &  \setlength{\fboxsep}{2pt}\fbox{\ensuremath{1}}
\end{pmatrix},\quad B=\begin{pmatrix}
{\fbox{\ensuremath{B_0}}} &  & &  & \\
 & {\fbox{\ensuremath{B_1}}} &  & & \\
  &  & \ddots & & \\
  &  & & \setlength{\fboxsep}{5pt}\fbox{\ensuremath{B_{k-1}}} & \\
  &  & &  &  \setlength{\fboxsep}{2pt}\fbox{\ensuremath{1}}
\end{pmatrix},
\end{equation}
and similarly
\begin{equation}
\fboxsep=6pt
C=\begin{pmatrix}
\setlength{\fboxsep}{2pt}\fbox{\ensuremath{1}}&  & &  &  \\
& {\fbox{\ensuremath{C_1}}} &  & &  \\
&   & \ddots &  \\
  &  & & \setlength{\fboxsep}{5pt}\fbox{\ensuremath{C_{k-1}}} & \\
  &  & &  &  {\fbox{\ensuremath{C_k}}}
\end{pmatrix},\quad D=\begin{pmatrix}
\setlength{\fboxsep}{2pt}\fbox{\ensuremath{1}}&  & &  &  \\
& {\fbox{\ensuremath{D_1}}} &  & &  \\
&   & \ddots &  \\
  &  & & \setlength{\fboxsep}{5pt}\fbox{\ensuremath{D_{k-1}}} & \\
  &  & &  &  {\fbox{\ensuremath{D_k}}}
\end{pmatrix}.
\end{equation}
We therefore have the following structure.

\begin{proposition}\label{prop:structure_of_projections}
    Let $A,B,C,D$ be orthogonal projections over $\mathbb C^{2k+1}$, satisfying
    \begin{equation}\label{eq:sum_4_proj}
        A+B+C+D=\frac{4k+4}{2k+1}\mathbb I_{2k+1}.
    \end{equation}
    Then, it is it possible to express them in a proper orthonormal basis as
    \begin{equation}\label{eq:a_b_projections}
        A = \bigoplus_{\ell=0}^{k-1} A_\ell \oplus E_{2k},\quad B = \bigoplus_{\ell=0}^{k-1} B_\ell \oplus E_{2k},
    \end{equation}
    and
    \begin{equation}\label{eq:c_d_equations}
         C = E_{0}\oplus\bigoplus_{\ell=1}^{k} C_\ell, \quad D = E_{0}\oplus\bigoplus_{\ell=1}^{k} D_\ell.
    \end{equation}
    Specifically, for $\ell=0,1,\dots, k-1$, $A_\ell$ and $B_\ell$ are given by Equation~\eqref{eq:operators_A_B}, and have support over $\ket{2\ell}\oplus \ket{2\ell+1}$, while for $\ell=0,1,\dots, k-1$, $C_\ell$ and $D_\ell$ are given by~\eqref{eq:operators_C_D} and have support over $\ket{2\ell-1}\oplus \ket{2\ell}$. Observe also that the entries of $A$ and $B$ corresponding to $\ket{2k+1}$, and those of $C$ and $D$ corresponding to $\ket{0}$, are equal to one.
\end{proposition}

\section{Self-testing of Slater state}\label{sec:self_testing_Slater}

We now introduce our strategy to self-test the measurements and, under certain conditions, the Slater state. Consider $n=d$ parties, such that each party has access to a Hilbert space $\mathcal H$ of dimension $d$. Suppose that the parties share a Slater state, which can be expressed in the standard basis $\{\ket{j}\}_{j=0,\dots,d-1}$ as in Equation~\eqref{eq:slater_state_canonical_basis}.

On each Hilbert space $\mathcal H$, there exists a family of projections $\{P_\mu\}_{\mu=0,1,\dots,N-1}$ summing to $x$ times the identity as in Equation~\eqref{eq:operators_sum_identity}, with $x=b/d$ an element in $\Lambda_N$. Correspondingly, each party has $N$ quantum measurements given by:
\begin{equation}\label{eq:local_measurements}
    \{P_\mu,\mathbb I_d-P_\mu\},\quad \mu=0,1,\dots,N-1.
\end{equation}
Define $\tilde E_{\mu,0}^{(k)}=P_\mu$ and $\tilde E_{\mu,1}^{(k)}=\mathbb I_d-P_\mu$, with $k\in[d]$ and $\mu\in[N]$. The corresponding quantum correlation is
\begin{equation}\label{eq:correlation}
\tilde{p}(\boldsymbol a|\boldsymbol \mu)=\mel{\Psi_\mathrm{S}}{\tilde E_{\boldsymbol \mu,\boldsymbol a}}{\Psi_\mathrm{S}},
\end{equation}
with $\boldsymbol a\in\{0,1\}^d$, $\boldsymbol{\mu}\in[N]^d$, and $\tilde E_{\boldsymbol \mu,\boldsymbol a}=\tilde E_{\mu_0,a_0}^{(0)}\otimes \tilde E_{\mu_1,a_1}^{(1)}\otimes\dots \otimes  \tilde E_{\mu_{d-1},a_{d-1}}^{(d-1)}$.
\begin{definition}\label{def:canonical_strategy}
    We refer to the canonical strategy for generating the correlation $\tilde p(\boldsymbol{a}|\boldsymbol{\mu})$ in Equation~\eqref{eq:correlation} as the one that utilizes the Slater state $\ket{\Psi_{\mathrm S}}$ along with the local measurements $\tilde E_{\boldsymbol \mu,\boldsymbol a}$ defined in~\eqref{eq:local_measurements}.
\end{definition}

In the following Subsection, we prove that the canonical quantum strategy~\ref{def:canonical_strategy} enables self-testing of the local measurements~\eqref{eq:local_measurements} for all families of $N$ projections summing up to the identity. In Subsection~\ref{sec:self_testing_state}, we provide sufficient conditions that also ensure self-testing of the states. We show that these conditions hold for certain classes of projections whose sum is proportional to the identity, particularly in the cases $N=d+1$ in Section~\ref{sec:rank_1_example} and $N=4$ in Section~\ref{sec:4_proj}.

\subsection{Self-testing of measurements}\label{sec:self_testing_measurement}

In this Section, we prove that the strategy presented in Equation~\eqref{def:canonical_strategy} self-tests the measurements in Equation~\eqref{eq:local_measurements}. We begin by introducing a special projection on the tensor product Hilbert space $\mathcal H^{\otimes d}$, which plays a key role in self-testing, along with some properties of the Slater state. First, given an operator of the form $A=A_0\otimes A_1\otimes\dots\otimes A _{d-1}$, where each $A_i\in\mathcal B(\mathcal H)$, we define the action of the symmetrizer $\mathcal S$ on $A$ as~\cite{Procesi2007}
\begin{equation}\label{eq:symmetric_vector}
    \mathcal S(A)=\frac{1}{d!}\sum_{\sigma\in\mathbb P_d} A_{\sigma(0)}\otimes A_{\sigma(1)}\otimes \dots \otimes A_{\sigma(d-1)}.
\end{equation}
The action of the symmetrizer can be extended to all elements of $\mathcal B(\mathcal H^{\otimes d})$ by linearity.
Observe that $\mathcal S$ is a projection (idempotent) on the set of operators, and its image is the symmetric algebra~$\mathfrak A_f$, as we formally prove in Appendix~\ref{app:symmetric_algebras}. The Slater state is the only state in $\mathcal H^{\otimes d}$ that satisfies the condition
\begin{equation}\label{eq:slater_fixing_property}
    \mathcal S( A)\ket{\Psi_\mathrm{S}}=\alpha \ket{\Psi_{\mathrm S}}\quad :\quad\alpha \in \mathbb C,
\end{equation}
for all $A$ of the form $A=A_0\otimes A_1\otimes\dots\otimes A _{d-1}$. This is because $\ket{\Psi_{\mathrm S}}$ is the only pure state in the permutation algebra $\mathfrak B$, which is the commutant of the symmetric algebra $\mathfrak A_f$, the $C^*$-algebra generated by all $\mathcal S(A)$, see Appendix~\ref{app:symmetric_algebras} for the details.

Given $N$ projections $\{P_\mu\}_{\mu=0}^{N-1}$ in~\eqref{eq:operators_sum_identity}, suppose that $\rank(P_\mu)=r=b/N<d$, and consider the other projections
\begin{equation}
    R_\mu=\underbrace{P_\mu\otimes P_\mu\otimes\dots\otimes P_\mu}_r\otimes \underbrace{(\mathbb I-P_\mu)\otimes \dots\otimes (\mathbb I-P_\mu)}_{d-r},\quad \mu=0,\dots, N-1.
\end{equation}
Define the operators:
\begin{equation}\label{symmetric_projections}
    S_\mu=\frac{d!}{r!(d-r)!}\mathcal S(R_\mu),\quad \mu=0,\dots,N-1.
\end{equation}

The normalization factor guarantees that permutations leading to the same configuration in~\eqref{eq:symmetric_vector} are counted only once. It is straightforward to verify that $S_\mu$ acts as a projection, since the projections in the sum~\eqref{eq:symmetric_vector} are mutually orthogonal.

%The following result holds.
\begin{proposition}\label{prop:stabilizing_slater_state}
    The Slater state $\ket{\Psi_{\mathrm S}}$ is in the range of all the projections $S_\mu$, i.e.,
\begin{equation}\label{vino}
    S_\mu \ket{\Psi_{\mathrm S}}=\ket{\Psi_{\mathrm S}}.
\end{equation}
\end{proposition}
\begin{remark}
    The property~\eqref{vino} implies that, if we know that the first $r$ components of $\ket{\Psi_{\mathrm S}}$ are in the subspace $P_\mu\mathcal H$, then all the residual $d-r$ components are in the orthogonal subspace $P_\mu\mathcal H^\perp$. This property is used in~\cite{Saha2024} in order to obtain a perfect strategy.
\end{remark}
\begin{proof}
From Equation~\eqref{eq:slater_fixing_property}, we infer that $S_\mu \ket{\Psi_{\mathrm S}}= \alpha\ket{\Psi_{\mathrm S}}$, with $\alpha$ either $0$ or $1$. It remains to show that $\alpha\neq1$. The operator $P_\mu$ can be written in terms of a complete relation:
\begin{equation}\label{eq:projection_expansion}
    P_\mu=\sum_{j=0}^{r-1} \ketbra{j},
\end{equation}
where $\{\ket{j}\}_{j=0,\dots,d-1}$ represents the standard basis, while the Slater state can be expressed as~\eqref{eq:slater_state_canonical_basis}, up to a phase~\eqref{eq:singlet_property}.

Observe that $S_\mu$ is the sum of positive operators. In order to check that $\alpha\neq 1$, it is sufficient to check that there is one element in the sum for which the expectation value is different from zero. For example
\begin{align}
    R_\mu\ket{\Psi_{\mathrm S}}&=P_\mu\otimes P_\mu \otimes \dots \otimes P_\mu \otimes (\mathbb I-P_\mu) \otimes \dots \otimes (\mathbb I-P_\mu) \ket{\Psi_{\mathrm S}} \nonumber\\&=\sum_{i_0=0}^{r-1}\dots\sum_{i_{r-1}=0}^{r-1}\sum_{i_{r}=r}^{d-1}\dots\sum_{i_{d-1}=r}^{d-1}\ketbra{i_0}\otimes\dots\otimes \ketbra{i_{r-1}} \otimes \ketbra{i_{r}}\otimes \dots \otimes \ketbra{i_{d-1}} \ket{\Psi_{\mathrm S}}\nonumber\\
    &=\sum_{i_1=0}^{r-1}\dots\sum_{i_{r-1}=0}^{r-1}\sum_{i_{r}=r}^{d-1}\dots\sum_{i_{d-1}=r}^{d-1}\sum_{\sigma \in \mathbb P_d} \frac{1}{\sqrt{d!}} \mathrm{sign}(\sigma) \ket{\sigma(0)\dots\sigma(d-1)}\left(\delta_{i_0,\sigma(0)}\dots\delta_{i_{d-1},\sigma(d-1)}\right).
\end{align}

A straightforward computation yields:
\begin{align}
    \mel{\Psi_{\mathrm S}}{R_\mu}{\Psi_{\mathrm S}}&=\bra{\Psi_{\mathrm S}}P_\mu \otimes \dots \otimes P_\mu \otimes (\mathbb I-P_\mu) \otimes \dots \otimes (\mathbb I-P_\mu) \ket{\Psi_{\mathrm S}}\nonumber\\ &=\sum_{i_0=0}^{r-1}\dots\sum_{i_{r-1}=0}^{r-1} \sum_{i_{r}=r}^{d-1}\dots\sum_{i_{d-1}=r-1}^{d-1}\sum_{\sigma,\sigma' \in \mathbb P_d} \Bigg[\frac{1}{{d!}} \mathrm{sign}(\sigma) \mathrm{sign}(\sigma') \left(\delta_{i_0,\sigma(0)}\dots\delta_{i_{d-1},\sigma(d-1)}\right)\nonumber\\&\hspace{180pt} \braket{\sigma'(0)\dots\sigma'(d-1)}{\sigma(0)\dots\sigma(d-1)}\Bigg]\nonumber\\
    &=\sum_{i_1=0}^{r-1}\dots\sum_{i_{r-1}=0}^{r-1}\sum_{i_{r}=r}^{d-1}\dots\sum_{i_{d-1}=r}^{d-1}\sum_{\sigma\in \mathbb P_d} \frac{1}{{d!}} \left(\delta_{i_0,\sigma(0)}\dots\delta_{i_{d-1},\sigma(d-1)}\right)=\frac{r!(d-r)!}{{d!}}
\end{align}
Hence, $\alpha\neq 0$ and thus $S_\mu\ket{\Psi_{\mathrm{S}}}=\ket{\Psi_{\mathrm{S}}}$.
\end{proof}

\medskip

The previous property implies the following corollary. 

\begin{corollary}[Synchronous correlation]\label{cor:synconous_property}
Given $A_1,A_2,\dots,A_d\in\mathcal B(\mathcal H)$,  suppose that there is an index $\mu=0,\dots, d-1$ such that either $\#\{k:A_k=P_\mu\}>r$ or $\#\{k:A_k=\mathbb I_d -P_\mu\}>d-r$, with $\#\{O\}$ being the cardinality of $O$. Then:
\begin{equation}
    A_1\otimes A_2\otimes \dots \otimes A_d \ket{\Psi_{\mathrm{S}}}=0.
\end{equation}
\end{corollary}

\begin{proof}
    It follows directly from the Proposition~\ref{prop:stabilizing_slater_state}, and the observation that
    \begin{equation}
        A_1\otimes A_2\otimes \dots \otimes A_d\, S_\mu=0,
    \end{equation}
    as in all terms appearing in the sum~\eqref{eq:symmetric_vector} there are exactly $r$  copies of $P_\mu$ and $d-r$ copies $\mathbb I_d - P_\mu$. Since $P_\mu(\mathbb I_d - P_\mu)=0$, the claim follows.
\end{proof}

\begin{remark}
    We called the previous corollary synchronous correlation as it it implies a synchronous property for the quantum correlation $\tilde p(\boldsymbol{a}|\boldsymbol{\mu})$. A bipartite correlation $p(a,a'|\mu,\mu')$ is synchronous if for all $\mu$ and $a\neq a'$ we have that $p(a,a'|\mu,\mu)$. In the multipartite case, the synchronous property~\ref{cor:synconous_property} implies that
    %\begin{equation}
        $\tilde p(\boldsymbol{a}|\boldsymbol{\mu})=0$
    %\end{equation}
    with $\boldsymbol{\mu}=(\mu,\mu,\dots,\mu)$ and $\boldsymbol{a}$ satisfying $\sum_{k=0}^{d-1} a_k\neq n-r$. Similarly to the bipartite case~\cite{Paulsen2016,Zhao2024}, the synchronous property carries significant implications that are foundational to the self-testing statement from this work. 
\end{remark}

From now on, we introduce the notation $\overline{\boldsymbol{\mu}}=(\mu,\mu,\dots,\mu).$
Using Theorem~\ref{th:mancinska} and Proposition~\ref {prop:stabilizing_slater_state}, we can proceed with the proof of self-testing for the measurement.

\begin{theorem}\label{th:self_testing_measurement}
    The quantum correlation $\tilde p(\boldsymbol a|\boldsymbol \mu)$ in Equation~\eqref{eq:correlation} obtained from the canonical strategy~\ref{def:canonical_strategy} self-tests the measurements~\eqref{eq:local_measurements}.
\end{theorem}

\begin{proof}
Consider a system of $d$ parties, each with a Hilbert space $\mathcal H_0,\mathcal H_1,\dots\mathcal H_{d-1}$. Suppose that they share a state
\begin{equation}\label{eq:generic_state}
    \ket{\Psi}\in \bigotimes_{k=0}^{d-1} \mathcal H_{k}
\end{equation}
and consider $N$ POVM elements in each subspace ${\mathcal H_{k}}$:
\begin{equation}\label{eq:generic_measurement}
    \{E_{\mu,0}^{(k)},E_{\mu,1}^{(k)}\}\subset\mathcal B(\mathcal H_k),\quad   \mu\in[N],\quad  k\in[d]
\end{equation}
with $E_{\mu,0}^{(k)}+E_{\mu,1}^{(k)}=\mathbb I_{\mathcal H_{k}}$, and $\mathbb I_{\mathcal H_{k}}$ being the identity on $\mathcal H_k$.
Suppose that these POVM and the state $\ket{\Psi}$ generate the same correlations $\tilde p$ as in Definition~\ref{def:canonical_strategy} and Equation~\eqref{eq:correlation}, that is
\begin{equation}\label{eq:correlation_test}
    \tilde{p}(\boldsymbol{a}|\boldsymbol{\mu})=\mel{\Psi}{E_{\boldsymbol\mu,\boldsymbol{a}}}{\Psi},
\end{equation}
with $\boldsymbol{a}\in\{0,1\}^d$, $\boldsymbol\mu\in[d]^n$, and
$E_{\boldsymbol\mu,\boldsymbol{a}}=E_{\mu_0,a_0}^{(k)}\otimes E_{\mu_1,a_{1}}^{(k)}\otimes\dots\otimes E_{\mu_{d-1},a_{d-1}}^{(d-1)}$.

The proof is split into three parts. First, in Part I, we prove that a synchronous property~\ref{cor:synconous_property} holds for $\ket{\Psi}$. In Part II, we show that all operators $E_{\mu,a}^{(k)}$ are projectors for all $\mu$, $k$ and $a$. This is a consequence of the synchronous property described in Part I. Finally, in Part III, we conclude that the sum of the projections $E_{\mu,0}^{(k)}$ is proportional to the identity multiplied by $x$. All of this supported by the Theorem~\ref{th:mancinska} proves the claim.

\bigskip

\textit{Part I}

Without loss of generality, we assume that, if $\rho_k$ is the reduced density matrix of $\ketbra{\Psi}$ to the space of $\mathcal H_k$, and $\rho_k$ has maximal support.

Following~\ref{cor:synconous_property}, we aim to prove that taking $A_0\in\mathcal B(\mathcal H_{0}),\dots,A_{d-1}\in\mathcal B(\mathcal H_{d-1})$, and assuming that for a fixed $\mu$ either $\#\{k:A_k=E^{(k)}_{\mu,0}\}>r$ or $\#\{k:A_k=E_{\mu,1}^{(k)}\}>d-r$, the following equality holds:
\begin{equation}\label{eq:terms_equal_zero}
    A_1\otimes A_2 \otimes \dots \otimes A_d \ket{\Psi}=0.
\end{equation}
Indeed, suppose $A_k=M_\mu^{(k)}$ for $k=\in[r]$, and $A_k=\mathbb I_{\mathcal H_k}$ for $k\in[r,n]$. Then, setting $\overline{\boldsymbol{\mu}}=(\mu,\mu,\dots\mu)$, we have
\begin{align}
    0=&\sum_{\boldsymbol a\in\{0\}^r\times\{0,1\}^{n-r}}\tilde p(\boldsymbol a|\overline{\boldsymbol{\mu}})=\sum_{\boldsymbol a\in\{0\}^r\times\{0,1\}^{n-r}}\mel{\Psi_{\mathrm{S}}}{\tilde E_{\overline{\boldsymbol{\mu}},\boldsymbol a}}{\Psi_{\mathrm{S}}}=\sum_{\boldsymbol a\in\{0\}^r\times\{0,1\}^{n-r}}\mel{\Psi}{E_{\overline{\boldsymbol{\mu}},\boldsymbol a}}{\Psi}\nonumber\\ &=\mel{\Psi}{\bigotimes_{k\in[r+1]}{E_{\overline{\boldsymbol{\mu}},0}^{(k)}\otimes \bigotimes_{k\in[r+1,d]} \mathbb I_{\mathcal H_{k}}}}{\Psi}=\Bigg\Vert{\bigotimes_{k\in[r+1]}{E_{\overline{\boldsymbol{\mu}},0}^{(k)\frac{1}{2}}\otimes \bigotimes_{k\in[r+1,d]} \mathbb I_{\mathcal H_{k}}}\ket{\Psi}}\Bigg\Vert^2,
\end{align}
which implies that
\begin{equation}
    E_{\mu,0}^{(0)}\otimes E_{\mu,0}^{(1)}\otimes \dots \otimes E_{\mu,0}^{(r)}\otimes \mathbb I_{\mathcal H_{r+1}}\otimes\dots \otimes \mathbb I_{\mathcal H_{d-1}}\ket{\Psi}=0.
\end{equation}
%Equation~\eqref{eq:terms_equal_zero} can be derived similarly for all possible conditions mentioned above.

\bigskip

\textit{Part II}

We now show that all the operators $E_{\mu,0}^{(k)}$ are projections. We proceed with $k=0$, and begin with the following observation, which is a consequence of Equation~\eqref{eq:terms_equal_zero}. Consider
\begin{align}\label{tintodeverano}
E_{\mu,0}^{(0)}\otimes \left(\bigotimes_{k\in[1,d]}\mathbb I_{\mathcal {H}_{k}}\right)\ket{\Psi}&=E_{\mu,0}^{(0)}\otimes\bigotimes_{k\in[1,d]} (E_{\mu,0}^{(k)}+E_{\mu,1}^{(k)}) \ket{\Psi}=
\sum_{\boldsymbol a\in\{0\}\times \{0,1\}^{d-1}}E_{\overline{\boldsymbol{\mu}},\boldsymbol a}\ket{\Psi} \nonumber\\
&=\sum_{\substack{\boldsymbol a\in\{0\}\times \{0,1\}^{d-1}\\\sum_{k}a_k=d-r}}E_{\overline{\boldsymbol{\mu}},\boldsymbol a}\ket{\Psi}=\mathbb I_{\mathcal H_0}\otimes \sum_{\substack{\boldsymbol a\in\times \{0,1\}^{d-1}\\\sum_{k}a_k=d-r}}E_{\overline{\boldsymbol{\mu}},\boldsymbol a}^{[1,d]}\ket{\Psi}.
\end{align}
In the previous equation, we introduced, for $m>n$:
\begin{equation}
E_{\boldsymbol \mu,\boldsymbol a}^{[n,m]}:=E_{\mu_0,a_0}^{(n)}\otimes E_{\mu_1,a_1}^{(n+1)}\otimes \dots \otimes E_{\mu_{m-n-1},a_{m-n-1}}^{(m-1)}
\end{equation}
If the sum in the first line was expanded, the only non-zero terms would have $r$ elements equal to $E_{\mu,0}^{(k)}$ and $d-r$ elements equal to $E_{\mu,1}^{(k)}$. Furthermore, in the last line, we can exchange $E_{\mu,0}^{(0)}$ with the identity $\mathbb I_{\mathcal H_0}$, as no contribution would come from $E_{\mu,1}^{(0)}$.
Using the condition~\eqref{tintodeverano}, we can explicitly show that the operator $E_{\mu,0}^{(0)}$ is a projection. Take $\ket{\varphi}\in\mathcal H_{0}$, and we claim that $(E_{\mu,0}^{(0)})^2\ket{\varphi}=E_{\mu,0}^{(0)}\ket{\varphi}$. Explicitly
\begin{align}\label{eq:passage_of_M}
    (E_{\mu,0}^{(0)})^2\ket{\varphi}&=(E_{\mu,0}^{(0)})^2\rho_0\rho_0^{-1}\ket{\varphi}=(E_{\mu,0}^{(0)})^2\tr_{[1,d]}({\ketbra{\Psi}})\rho_0^{-1}\ket{\varphi}\nonumber\\
    &=E_{\mu,0}^{(0)}\tr_{[1,d]}(E_{\mu,0}^{(0)}\otimes \mathbb I_{\mathcal H_{1}}\otimes \dots\otimes\mathbb I_{\mathcal H_{d-1}}{\ketbra{\Psi}})\rho_1^{-1}\ket{\varphi}\nonumber\\
    &=\sum_{\substack{\boldsymbol a\in\times \{0,1\}^{d-1}\\\sum_{k}a_k=d-r}}E_{\mu,0}^{(0)}\tr_{[1,d]}(\mathbb I_{\mathcal H_0}\otimes E_{\overline{\boldsymbol{\mu}},\boldsymbol a}^{[1,d]}{\ketbra{\Psi}})\rho_1^{-1}\ket{\varphi}\nonumber\\
    &=\sum_{\substack{\boldsymbol a\in\times \{0,1\}^{d-1}\\\sum_{k}a_k=d-r}}\tr_{[1,d]}(E_{\mu,0}^{(0)}\otimes E_{\overline{\boldsymbol{\mu}},\boldsymbol a}^{[1,d]}{\ketbra{\Psi}})\rho_1^{-1}\ket{\varphi}\nonumber\\
    &=\tr_{[1,d]}(E_{\mu,0}^{(0)}\otimes \mathbb I_{\mathcal H_{1}}\otimes \dots\otimes\mathbb I_{\mathcal H_{d-1}}{\ketbra{\Psi}})\rho_1^{-1}\ket{\varphi}=E_{\mu,0}^{(0)}\ket{\varphi}.
\end{align}
%\rev{red: I could write both %the partial trace and the %identity $\mathbb I_{[n,m]}$.} 

In the last equality we introduced the partial trace
\begin{equation}
    \tr_{[n,m]}(A):=\tr_{\mathcal H_n,\mathcal H_{n+1},\dots,\mathcal H_{m-1}}(A).
\end{equation}
Notice that the assumption on $\rho_0$ having full support is necessary in order to define $\rho_0^{-1}$ in the first line. 
We can repeat the same arguments for all parties so that 
\begin{equation}
    E_{\mu,0}^{(k)}=(E_{\mu,0}^{(k)})^2=(E_{\mu,0}^{(k)})^*,\quad E_{\mu,1}^{(k)}=(E_{\mu,1}^{(k)})^2=(E_{\mu,1}^{(k)})^*.
\end{equation}
hold for all $k=0,\dots,d-1$ and $\mu=0,\dots,N-1$.
Consequently, all the $\{E_{\mu,0}^{(k)},E_{\mu,1}^{(k)}\}$ are projective measurement.

\bigskip

\textit{Part III}

It remains to show that the projections $E_{\mu,0}^{(k)}$ sum up to the identity. We mimic the proof from~\cite{Mancinska2024}. For $k=0,\dots, N-1$, define $E^{(k)}=\sum_{\mu=0}^{N-1} E_{\mu,0}^{(k)}\in\mathcal B(\mathcal H_k)$. We claim  that $E^{(k)}=x\mathbb I_{\mathcal H_{k}}$. For $k=0$ one has
\begin{equation}
    \mathrm{tr}(\rho_0 E^{(0)})=\mel{\Psi}{E^{(0)}\otimes \mathbb I_{\mathcal H_1}\otimes \dots \otimes \mathbb I_{\mathcal H_{d-1}}}{\Psi}=x.
\end{equation}
Similarly, using~\eqref{eq:passage_of_M},
\begin{align}
    \mathrm{tr}{(\rho_0 (E^{(0)})^2)}&=\mel{\Psi}{(E^{(0)})^2\otimes \mathbb I_{\mathcal H_{1}}\otimes \dots \otimes \mathbb I_{\mathcal H_{d-1}}}{\Psi}=\sum_{\mu\nu=0}^{N-1} \mel{\Psi}{E_{\nu,0}^{(0)}E_{\mu,0}^{(0)}\otimes \mathbb I_{\mathcal H_1}\otimes \dots \otimes \mathbb I_{\mathcal H_{d-1}}}{\Psi}\nonumber\\
    &=\sum_{\mu\nu=0}^{N-1}\mel{\Psi}{E_{\nu,0}^{(0)}\otimes \sum_{\substack{\boldsymbol a\in\times \{0,1\}^{d-1}\\\sum_{k}a_k=d-r}} E_{\overline{\boldsymbol{\mu}},\boldsymbol a}^{[1,d]}}{\Psi}=\sum_{\mu\nu=0}^{N-1}\mel{\Psi_{\mathrm S}}{\tilde E_{\nu,0}^{(0)}\otimes \sum_{\substack{\boldsymbol a\in\times \{0,1\}^{d-1}\\\sum_{k}a_k=d-r}} \tilde E_{\overline{\boldsymbol{\mu}},\boldsymbol a}^{[1,d]}}{\Psi_{\mathrm S}}\nonumber\\
    &=\sum_{\mu\nu=0}^{N-1}\mel{\Psi_{\mathrm S}}{\tilde E_{\nu,0}^{(0)}\tilde E_{\mu,0}^{(0)}\otimes \mathbb I_{d}\otimes \dots \otimes \mathbb I_{d}}{\Psi_{\mathrm S}}=x^2.
\end{align}
Observe that, in the last line, since the correlation generated by $E_{\overline{\boldsymbol{\mu}},\boldsymbol a}$ and the state $\ket{\Psi}$ is the same $\tilde p(\boldsymbol{a}|\overline{\boldsymbol{\mu}})$ as the one generated by $\tilde E_{\overline{\boldsymbol{\mu}},\boldsymbol a}$ and $\ket{\Psi_\mathrm S}$.
In the last line, we utilized the fact that we are obtaining a sum that is linear in the quantum correlation $\tilde p$, thereby yielding the same result we would obtain for the Slater state. But then
\begin{equation}
    \mathrm{tr}(\rho_0 E^{(0)})^2= \mathrm{tr}(\rho_0 E^{(0)2})=x^2.
\end{equation}
The Cauchy-Schwarz inequality implies that 
\begin{equation}
    E^{(0)}=\sum_{\mu} E_{\mu,0}^{(0)}=x\mathbb I_{\mathcal H_0}.
\end{equation}
The same proof can be repeated for all $k$. From Theorem~\ref{th:mancinska} each Hilbert space $\mathcal H_{k}$ can be written as a tensor product
\begin{equation}\label{eq:tensor_product_local}
    \mathcal H_{k}=\mathcal H\otimes \mathcal H_{k}',
\end{equation}
with $\mathcal H=\mathbb C^d$, and each $E_{\mu,0}^{(i)}$ decomposes as
\begin{equation}\label{eq:isometry_measurements}
    E_{\mu,0}^{(k)}=P_\mu \otimes \mathbb I_{\mathcal H_{k}'},
\end{equation}
which ends the proof.
%is the original claim, see Equation~\eqref{eq:self-testing_measurement}
\end{proof}

\subsection{Self-testing of the state}\label{sec:self_testing_state}

We continue to prove the self-testing statement for the state itself using the projections $S_\mu$ defined in~\eqref{symmetric_projections}. We propose the following conjecture, which is a stronger version of proposition~\ref{prop:stabilizing_slater_state}.
\begin{conjecture}\label{prop:conjecture}
    The Slater state $\ket{\Psi_{\mathrm S}}$ is the only state in the ranges of $S_\mu$.
\end{conjecture}

With this conjecture, we are now equipped to establish the full self-testing result.

\begin{theorem}\label{th:self_testing_state}
    If Conjecture~\ref{prop:conjecture} holds, then the quantum correlation $\tilde p(\boldsymbol a|\boldsymbol \mu)$ in Equation~\eqref{eq:correlation} defined via the quantum strategy~\ref{def:canonical_strategy} self-tests the Slater state~\eqref{eq:slater_state_canonical_basis}.
\end{theorem}
\begin{proof}
    Similarly, as in the proof of Theorem~\ref{th:self_testing_measurement}, let $d$ parties $\mathcal H_0,\dots,\mathcal H_{d-1}$ share a state in the form~\eqref{eq:generic_state}, and let each system possesses $N$ measurements as in Equation~\eqref{eq:generic_measurement}, generating the quantum correlation~\eqref{eq:correlation}. Due to Theorem~\ref{th:self_testing_measurement}, each  Hilbert space $\mathcal H_{k}$ decomposes as in Equation~\eqref{eq:tensor_product_local}, with $\mathcal H=\mathbb C^d$, and each $E_{\mu,0}^{(i)}$ can be written as~\eqref{eq:isometry_measurements}. After tracing out all the Hilbert spaces $\mathcal H_{k}'$, only $\mathcal H^{\otimes d}$ is left. Let $\rho$ denote the reduced density matrix of $\ketbra{\Psi}$ over $\mathcal H^{\otimes d}$. It is straightforward to check that
    \begin{equation}
        S_\mu=\sum_{\substack{\boldsymbol a\in\times \{0,1\}^{d-1}\\\sum_{k}a_k=d-r}}\tilde E_{\overline{\boldsymbol{\mu}}|\boldsymbol a}.
    \end{equation}
    For $\mu=0,\dots,N-1$ it holds:
    \begin{align}
        \tr(\rho S_\mu)&=\sum_{\substack{\boldsymbol a\in\times \{0,1\}^{d-1}\\\sum_{k}a_k=d-r}}\mathrm{tr}(\rho\tilde E_{\overline{\boldsymbol{\mu}}|\boldsymbol a})= \sum_{\substack{\boldsymbol a\in\times \{0,1\}^{d-1}\\\sum_{k}a_k=d-r}}\tilde p(\boldsymbol a|\overline{\boldsymbol{\mu}})=\sum_{\substack{\boldsymbol a\in\times \{0,1\}^{d-1}\\ \sum_{k}a_k=d-r}}\mel{\Psi_{\mathrm{S}}}{\tilde E_{\overline{\boldsymbol{\mu}}|\boldsymbol a}}{\Psi_{\mathrm{S}}}\nonumber\\
        &=\mel{\Psi_{\mathrm S}}{S_\mu}{\Psi_{\mathrm S}}=1.
    \end{align}
    As a consequence, $\mathrm{supp}({\rho})\subset \mathrm{range} (S_\mu)$. But then, from Conjecture~\ref{prop:conjecture}, it follows that that $\rho=\ketbra{\Psi_{\mathrm{S}}}$.
\end{proof}

\bigskip

Next, we introduce a second conjecture, which implies the first one. This constitutes a stronger assertion, however, it can be directly verified in specific cases.

\begin{conjecture}\label{prop:conjecture_3}
    For $\mu=0,\dots,N-1$, define the symmetric operators
    \begin{align}\label{eq:T_operators}
        T_\mu &=d\,\mathcal S(P_\mu\otimes\mathbb I_d\otimes \dots\otimes\mathbb I_d)\nonumber\\ &=P_\mu\otimes \mathbb I_d\otimes\dots\otimes \mathbb I_d+\mathbb I_d \otimes P_\mu\otimes\dots\otimes \mathbb I_d +\dots+\mathbb I_d\otimes\dots \otimes \mathbb I_d  \otimes P_\mu.
    \end{align}
    Then, $T_\mu$ generate all the symmetric algebra $\mathfrak A_f$. Equivalently, the commutant of the set
    \begin{equation}\label{eq:set_of_t}
        \mathcal K=\{T_\mu:\mu=0,\dots,N-1\}
    \end{equation}
    is $\mathfrak B$, the permutation algebra from Appendix~\ref{app:symmetric_algebras}:
    \begin{equation}\label{eq:commutant_condition}
        \mathcal K'=\mathfrak B.
    \end{equation}
\end{conjecture}
Observe that, due von Neumann's bicommutant Theorem, and the fact that $\sum_{\mu}T_\mu=xd\mathbb\, I_\mathcal H$, Equation~\eqref{eq:commutant_condition} is equivalent to $\mathcal K''=\mathfrak A_f$, which means that $T_\mu$ generates entire symmetric algebra.

%The second conjecture is stronger than the first one.
\begin{proposition}
    Suppose Conjecture~\ref{prop:conjecture_3} holds. Then, Conjecture~\ref{prop:conjecture} holds.
\end{proposition}
\begin{proof}
    Let $S_\mu\ket{\Psi}=\ket{\Psi}$ for all $\mu=0,\dots,N-1$. Then
    \begin{align}
        T_\mu S_\mu=r S_\mu.
    \end{align}
    Indeed, take a general element in the decomposition~\eqref{eq:symmetric_vector} of $S_\mu$. Applying $T_\mu$ to this element results in $r$ identical copies, which establishes the stated relation. Thus, $T_\mu \ket{\Psi}=T_\mu S_\mu \ket{\Psi}=r\ket{\Psi}$, and thus $\ketbra{\Psi}$ is an element in the commutant of $\{Z_\mu:\mu=0,1,\dots,N-1\}$. This implies that $\ketbra{\Psi}$ is a rank-1 projection in $\mathfrak B$, which means that $\ketbra{\Psi}=\ketbra{\Psi_{\mathrm S}}$, as $\ketbra{\Psi_{\mathrm S}}$ is the only rank-1 projection in $\mathfrak B$.
\end{proof}

\bigskip

This Section can be summarized with the following relation
%\begin{tcolorbox}
\begin{center}
Conjecture~\ref{prop:conjecture_3}$\implies$ Conjecture~\ref{prop:conjecture}$\implies$ Self-testing of the state.
\end{center}
%\end{tcolorbox}

\subsection{Testing the Conjecture}\label{Section:Algebraic_approach}

We prove Conjecture~\ref{prop:conjecture_3} for the cases $N=d+1$ and $N=4$. Because the operators $T_\mu$ defined in~\eqref{eq:T_operators} are symmetric, one immediately obtains the set inclusions:
\begin{align}
    \mathcal K'&\supseteq\mathfrak B,\\
    \mathcal K&\subseteq\mathcal K''\subseteq\mathfrak A_f.
\end{align}
Thus, to prove Conjecture~\ref{prop:conjecture_3}, it suffices to show that either the first inclusion or the third inclusion is actually an equality. Furthermore, since the span of $K$ contains the identity, we have that $\mathcal K''$ is also the $C^*$-algebra generated by $\mathcal K$.

Define the following set
\begin{equation}
    \mathscr{LW}=\{X\in\mathcal M_d: \mathcal S(X\otimes \mathbb I_d\otimes\dots \otimes\mathbb I_d)\in \mathcal K''\},
\end{equation}
with $\mathcal S$ the symmetrizer (see Equation~\eqref{eq:symmetric_vector}). It is evident that $\mathscr{LW}$ forms a vector space. Moreover, for any $X, Y\in\mathscr{LW}$, the commutator $[X,Y]$ also lies in $\mathscr{LW}$. Consequently, $\mathscr{LW}$ is a complex Lie subalgebra of $\mathcal M_d$. In addition, all the operators $P_\mu$ as well as the identity $\mathbb I_d$, belong to $\mathcal M_d$. Finally, the commutant of $\mathscr{LW}$ is trivial, and it is given by $\mathbb C\mathbb I_d$.

The following result holds.
\begin{lemma}
    Suppose that
    %\begin{equation}
        $\mathscr{LW}=\mathcal M_d.$
    %\end{equation}
    Then, $\mathcal K''=\mathfrak A_f$.
\end{lemma}
\begin{proof}
    It is equivalent to the fact that the Lie group generated by $\mathcal M_d$ is the whole $\mathrm{GL}(\mathbb C,d)$~\cite{Hall2015}. It is enough to prove that, given any operator $X$ in $\mathcal M_d$, $X^{\otimes d}$ belongs to $\mathcal K''$. To this end, take unitary $X$ and write it as
    \begin{equation}
        X = \mathrm e^{\mathrm i H},
    \end{equation}
    with $H$ a Hermitian operator. Then
    \begin{equation}
        X^{\otimes d}=\mathrm  e^{\mathrm i H}\otimes\dots\otimes \mathrm  e^{\mathrm i H}=\mathrm{exp}({\mathrm i d\,\mathcal S(H\otimes\mathbb I_d\otimes\dots\otimes\mathbb I_d)})\in\mathcal K'',
    \end{equation}
    with $\mathcal S$ the symmetrizer~\eqref{eq:symmetric_vector}.
    A similar proof can be repeated for $X$ being a strictly positive operator, as 
    it can always be written as $X=\mathrm e^{H}$ and proceed as in the unitary case. Furthermore, if $X$ is positive, it can be seen as a limit of strictly positive operators, so $X^{\otimes d}\in\mathcal K''$. Since any operator is the product of a positive and a unitary operators (via polar decomposition), we conclude that $\mathcal K''=\mathfrak A_f$.
\end{proof}

\medskip

The following two theorems support the validity of conjecture~\ref{prop:conjecture_3}.
\begin{theorem}\label{th:conjecture_one_dimensional_projections}
    Let $N=d+1$ in the quantum strategy defining the correlation $\tilde p(\boldsymbol a|\boldsymbol{\mu})$ in Eq.~\ref{def:canonical_strategy}. Then, $\mathscr{LW}$ is the whole $\mathcal M_d$.
\end{theorem}
\begin{theorem}\label{th:conjecture_four_projections}
    Let $N=4$ in the quantum strategy defining the correlation $\tilde p(\boldsymbol a|\boldsymbol{\mu})$ in Eq.~\ref{def:canonical_strategy}. Then, $\mathscr{LW}$ is the whole $\mathcal M_d$.
\end{theorem}
We prove Theorem~\ref{th:conjecture_one_dimensional_projections} for the case of $N=d+1$ discussed in Section~\ref{sec:rank_1_example}. The proof for the case $N=4$ is lengthy, and it is postponed to appendix~\ref{app:4_proj}.

\begin{proof}[Proof of~\ref{th:conjecture_one_dimensional_projections}]
Consider the $N=d+1$ projectors discussed in Section~\ref{sec:rank_1_example}. They are one-dimensional projections, and can be written as
\begin{equation}\label{eq:rank_1_projector}
    P_\mu=\ketbra{\psi_\mu},\quad \mu=0,\dots,d.
\end{equation}
The Lie algebra $\mathscr{LW}$ contains the commutant 
\begin{align}\label{eq:commutator_pmu}
[P_\mu,P_\nu]&=\frac{1}{d}\left(\ketbra{\psi_\nu}{\psi_\mu}-\ketbra{\psi_\mu}{\psi_\nu}\right)\\
[P_\mu,[P_\mu,P_\nu]]&= - \frac{1}{d} \left(\ketbra{\psi_\nu}{\psi_\mu}+\ketbra{\psi_\mu}{\psi_\nu}\right) -\frac{2}{d^2}\ketbra{\psi_\mu}\label{eq:second_commutator_pmu}.
\end{align}
Writing the linear combinations between~\eqref{eq:rank_1_projector},~\eqref{eq:commutator_pmu} and~\eqref{eq:second_commutator_pmu}, we get that for all $\mu,\nu=0,\dots,d$:
\begin{align}
    \ketbra{\psi_\mu}{\psi_\nu},\ketbra{\psi_\nu}{\psi_\mu}\in\mathscr{LW}.
\end{align}
Since $\{\ket{\psi_\mu}\}_{\mu=0,\dots,d-1}$ span all $\mathbb C^d$, we have that $\{\ketbra{\psi_\mu}{\psi_\nu}\}_{\mu,\nu=0,\dots,d-1}$ span all $\mathcal M_d$, and thus $\mathscr{LW}=\mathcal M_d$.
\end{proof}

\section{Robustness}\label{sec:self-testing_robust}

It is possible to prove that all the results regarding self-testing presented in Section~\ref{sec:self_testing} are actually robust.

\begin{theorem}[Robustness]\label{th:self-testing_robust}
    The self-test statements~\ref{th:self_testing_measurement} and~\ref{th:self_testing_state} are all robust.
\end{theorem}

The proof is a straightforward generalization of the one presented in~\cite{Mancinska2024}, and it is based on the following result shown in the aforementioned article.
\begin{theorem}\label{th:gowers_hatami_algebras}
Let $N\geqslant 3$, and $x\in\Lambda_N$, with $x=b/d$ in lowest terms. Let $\{P_\mu\}$ be orthogonal projections satisfying~\eqref{eq:operators_sum_identity}. Then, for any $\varepsilon >0$, there exists $\delta>0$ and $m\in\mathbb N$ such that, for all $r\in\mathbb N$, any density matrix $\rho\in\mathcal M_r$ and any positive $M_0,\dots,M_{N-1}\in\mathcal M_r$ with norm less than $1$ that satisfy the following three properties:
\begin{itemize}
    \item[\textnormal{(a)}] $\norm{{M_\mu-M_\mu^2}}_\rho\leqslant\delta$, with $\mu=0,\dots N-1$,
    \item[\textnormal{(b)}] $\norm*{x\mathbb I_r-\sum_{\mu=0}^{N-1}M_\mu}_\rho\leqslant\delta$,
    \item[\textnormal{(c)}] $\abs*{\mathrm{Tr}(\rho (W_1W_2-W_2W_1)}\leqslant \delta$ for monomials $W_1$ and $W_2$ of degree at most $m$ in $M_0,\dots,M_{N-1}\in\mathcal M_r$,
\end{itemize}
there exist $s\in\mathbb N$ and an isometry $V:\mathbb C^r\rightarrow\mathbb C^d\otimes \mathbb C^s$ such that for all $\mu=0,\dots,N-1$ we have $\norm*{M_\mu-V^*(P_\mu\otimes \mathbb I_s)V}_\rho\leqslant \varepsilon.$
\end{theorem}
We shall use the norm $\norm{X}_\rho=\Tr(\rho X^*X)^{1/2}$ introduced in the previous theorem. In what follows, we consider a quantum strategy $p$ generated by $d$ parties, each associated with a Hilbert space $\mathcal H_0,\mathcal H_1,\dots,\mathcal H_{d-1}$, respectively. The parties share a quantum state
\begin{equation}\label{eq:generic_state_2}
    \ket{\Psi}\in \bigotimes_{k=0}^{d-1} \mathcal H_{k}
\end{equation}
and $N$ POVM, in each subspace ${\mathcal H_{k}}$, of the form:
\begin{equation}\label{eq:generic_measurement_2}
    \{E_{\mu,0}^{(k)},E_{\mu,0}^{(k)}\}\subset\mathcal B(\mathcal H_k),\quad   \mu\in[N],\quad  k\in[d].
\end{equation}
They generate a correlation
\begin{equation}\label{eq:correlation_test_2}
    {p}(\boldsymbol a|\boldsymbol\mu)=\mel{\Psi}{E_{\boldsymbol \mu,\boldsymbol a}}{\Psi},\quad \boldsymbol a\in\{0,1\}^d, \quad \boldsymbol \mu\in[N]^d.
\end{equation}
Assume that this correlation is $\delta$-close with respect to the norm $1$ to the canonical correlation $\tilde p(\boldsymbol a|\boldsymbol {\mu})$ in accordance with Definition~\ref{def:canonical_strategy}, that is,
\begin{equation}\label{eq:closeness_p_canonical}
    \norm{p-\tilde p}_1=\sum_{\boldsymbol a\in\{0,1\}^d}\sum_{\boldsymbol \mu\in[N]^d}\abs{p(\boldsymbol a|\boldsymbol\mu)-\tilde p(\boldsymbol a|\boldsymbol\mu)}\leqslant\delta.
\end{equation}
To continue, we shall prove three additional lemmas, which correspond to Lemma 3.5 (a), Lemma 3.7 and Lemma 6.3 of~\cite{Mancinska2024}.

\begin{lemma}[Approximate synchronous correlation]\label{lem:approx_sync_cor}
    Let $N\geqslant 3$ and $x\in\Lambda_N$, and let $p(\boldsymbol a|\boldsymbol \mu)$ be a quantum correlation of the form~\eqref{eq:correlation_test_2} generated by the state $\ket{\Psi}$~\eqref{eq:generic_state_2} and the measurements $\{E_{\mu,0}^{(k)},E_{\mu,0}^{(k)}\}$~\eqref{eq:generic_measurement_2}, and also suppose that they are $\delta$ close to $\tilde p(\boldsymbol a|\boldsymbol \mu)$ as in Equation~\eqref{eq:closeness_p_canonical}. Then, according to~\eqref{eq:correlation_test_2},
    \begin{equation}
        \bigg\lVert {\bigg(E_{\mu,0}^{(0)}\otimes \bigotimes_{i=1}^{d-1}\mathbb I_{\mathcal H_i} -\mathbb I_{\mathcal H_{0}}\otimes \sum_{\substack{\boldsymbol a\in\{0,1\}^{d-1}\\\sum_{k}a_k=d-r}} E_{\overline{\boldsymbol{\mu}},\boldsymbol a}^{[1,d]}\bigg)\ket{\Psi} }\bigg\rVert \leqslant\sqrt{\delta},
    \end{equation}
    and
    \begin{equation}
        \bigg\lVert {\bigg(E_{\mu,0}^{(0)}\otimes \bigotimes_{i=1}^{d-1}\mathbb I_{\mathcal H_i} -E_{\mu,0}^{(0)}\otimes \sum_{\substack{\boldsymbol a\in\{0,1\}^{d-1}\\\sum_{k}a_k=d-r}} E_{\overline{\boldsymbol{\mu}},\boldsymbol a}^{[1,d]}\bigg)\ket{\Psi} }\bigg\rVert \leqslant\sqrt{\delta},
    \end{equation}
    for a fixed $\mu$. Similar properties hold for different $k$ in $E_{\mu,0}^{(k)}$.
\end{lemma}
Observe that, in the exact case $\delta=0$, Lemma~\ref{lem:approx_sync_cor} is a consequence of Corollary~\ref{cor:synconous_property}.
\begin{proof}
    Observe that
    \begin{equation}
        Q:=\sum_{\substack{\boldsymbol a\in \{0,1\}^{d-1}\\\sum_{k}a_k=d-r}} E_{\overline{\boldsymbol{\mu}},\boldsymbol a}^{[1,d]} \leqslant\bigotimes_{i=1}^{d-1}\mathbb I_{\mathcal H_i}.
    \end{equation}
    Then
    \begin{align}
        &\bigg\lVert(E_{\mu,0}^{(0)}\otimes \bigotimes_{i=1}^{d-1}\mathbb I_{\mathcal H_i}-\mathbb I_{\mathcal H_0}\otimes Q)\ket{\Psi}\bigg\rVert^2 =\mel*{\Psi}{E_{\mu,0}^{(0)2}\otimes \bigotimes_{i=1}^{d-1}\mathbb I_{\mathcal H_i}}{\Psi}+ \mel{\Psi}{\mathbb I_{\mathcal H_0}\otimes Q^2}{\Psi}-2\mel{\Psi}{E_{\mu,0}^{(0)}\otimes Q}{\Psi}\nonumber\\
        &\leqslant  \mel*{\Psi}{E_{\mu,0}^{(0)}\otimes \bigotimes_{i=1}^{d-1}\mathbb I_{\mathcal H_i}}{\Psi}+ \mel{\Psi}{\mathbb I_{\mathcal H_0}\otimes Q}{\Psi}-2\mel{\Psi}{E_{\mu,0}^{(0)}\otimes Q}{\Psi}\nonumber\\
        &=\sum_{{\boldsymbol a\in\{0\}\times \{0,1\}^{d-1}}} p(\boldsymbol a|\overline{\boldsymbol{\mu}}) + \sum_{\substack{\boldsymbol a\in \{0,1\}^{d}\\\sum_{k=1}^{d-1}a_k=d-r}}p(\boldsymbol a|\overline{\boldsymbol{\mu}})- 2\sum_{\substack{\boldsymbol a\in\{0\}\times \{0,1\}^{d-1}\\\sum_{k=1}^{d-1}a_k=d-r}}p(\boldsymbol a|\overline{\boldsymbol{\mu}})\leqslant \norm{p-\tilde p}_1\leqslant \delta.
    \end{align}
    In the last line, note that $\tilde {p}$ vanishes on the elements being summed, and their contribution does not affect the result. Similarly
    \begin{align}        
        &\bigg\lVert(E_{\mu,0}^{(0)}\otimes \bigotimes_{i=1}^{d-1}\mathbb I_{\mathcal H_i}-E_{\mu,0}^{(0)}\otimes Q)\ket{\Psi}\bigg\rVert^2 = \mel{\Psi}{E_{\mu,0}^{(0)2}\otimes(\bigotimes_{i=1}^{d-1} \mathbb I_{\mathcal H_i} - Q)^2}{\Psi}\nonumber\\ &\leqslant \mel{\Psi}{E_{\mu,0}^{(0)}\otimes(\bigotimes_{i=1}^{d-1} \mathbb I_{\mathcal H_i} - Q)}{\Psi}=\mel*{\Psi}{E_{\mu,0}^{(0)}\otimes \bigotimes_{i=1}^{d-1}\mathbb I_{\mathcal H_i}}{\Psi} - \mel{\Psi}{E_{\mu,0}^{(0)}\otimes Q}{\Psi} \nonumber \\ &=\sum_{{\boldsymbol a\in\{0\}\times \{0,1\}^{d-1}}} p(\boldsymbol a|\overline{\boldsymbol{\mu}}) -\sum_{\substack{\boldsymbol a\in\{0\}\times \{0,1\}^{d-1}\\\sum_{k=1}^{d-1}a_k=d-r}}p(\boldsymbol a|\overline{\boldsymbol{\mu}})\leqslant \norm{p-\tilde p}_1\leqslant \delta
    \end{align}
\end{proof}

\begin{lemma}\label{lem:approx_sum_operators}
    Let $N\geqslant 3$ and $x\in\Lambda_N$, and let $p$ be a quantum correlation in the form~\eqref{eq:correlation_test_2} generated by the state~\eqref{eq:generic_state_2} and the measurements~\eqref{eq:generic_measurement_2}, and suppose that they are $\delta$ close to $\tilde p(\boldsymbol a|\boldsymbol \mu )$ as in Equation~\eqref{eq:closeness_p_canonical}. Then, for all $\mu\in[N]$ and $k\in[d]$:
    \begin{equation}
       \norm{x\mathbb I_{\mathcal H_k}-\sum_{\mu=0}^{N-1}E_{\mu,0}^{(k)}}_\rho < C{\delta}^{1/4}.
    \end{equation}
    with
    \begin{equation}
        C=(1+2x)\sqrt{\delta} +N^2.
    \end{equation}
\end{lemma}
\begin{proof}
    We begin with $k=0$. Let $p_0(a_0|\mu)$ and $\tilde p_0(a_0|\mu)$ be the marginals of $p(\boldsymbol a|\boldsymbol \mu )$ and $\tilde p(\boldsymbol a|\boldsymbol \mu )$ with respect to the first system. Then,
    \begin{equation}
        \abs{\sum_{\mu=0}^{N-1}p_{0}(0|\mu)-x}=\abs{\sum_{\mu=0}^{N-1}(p_{0}(0|\mu)-\tilde p_{0}(0|\mu))}\leqslant \norm{p-\tilde p}_1\leqslant \delta,
    \end{equation}
    which implies that
    \begin{equation}
        \sum_{\mu=0}^{N-1}p_{0}(0|\mu)\geqslant x-\delta.
    \end{equation}
    Similarly,
    \begin{equation}
        \abs{\sum_{\mu\nu=0}^{N-1}p_{01}(0,0|\mu,\nu)-x^2}\leqslant \delta,
    \end{equation}
    with $p_{01}(a_0,a_1|\mu,\nu)$ being the marginals over the first two parties. Furthermore, using Lemma~\ref{lem:approx_sync_cor}
    \begin{align}
        \sum_{\mu,\nu=0}^{N-1}\mel*{\Psi}{E_{\mu,0}^{(0)}E_{\nu,0}^{(0)}\otimes\bigotimes_{k\in[1,d]} \mathbb I_{\mathcal H_k}}{\Psi}&\approx_{N^2\sqrt{\delta}} \sum_{\mu,\nu=0}^{N-1}\mel{\Psi}{E_{\mu,0}^{(0)}\otimes \sum_{\substack{\boldsymbol a\in \{0,1\}^{d-1}\\\sum_{k}a_k=d-r}} E_{\overline{\boldsymbol{\nu}},\boldsymbol a}^{[1,d]}}{\Psi}\nonumber\\
        &\approx_\delta \sum_{\mu,\nu=0}^{N-1}\mel{\Psi_{\mathrm S}}{\tilde E_{\mu,0}^{(0)}\otimes \sum_{\substack{\boldsymbol a\in \{0,1\}^{d-1}\\\sum_{k}a_k=d-r}} \tilde E_{\overline{\boldsymbol{\nu}},\boldsymbol a}^{[1,d]}}{\Psi_{\mathrm S}}= x^2.
    \end{align}
    In the last line, we could substitute the sum with $\ket{\Psi}$ and $E_{\boldsymbol\mu,\boldsymbol a}$ with $\ket{\Psi_{\mathrm S}}$ and $\tilde E_{\boldsymbol\mu,\boldsymbol a}$, as a consequence of~\eqref{eq:closeness_p_canonical}. Furthermore, we used the notation $a\approx_\varepsilon b\iff |a-b|<\varepsilon$.
    \begin{align}
        \bigg\lVert{x \mathbb I_{\mathcal H_k}-\sum_{\mu=0}^{N-1}E_{\mu,0}^{(k)}}\bigg\rVert_{\rho_0}^2&=x^2+\sum_{\mu,\nu=0}^{N-1}\mel*{\Psi}{E_{\mu,0}^{(0)}E_{\nu,0}^{(0)}\otimes \bigotimes_{k\in[1,d]}\mathbb I_{\mathcal H_k}}{\Psi}-2x\sum_{\mu=0}^{N-1}p_0(0|\mu)\\ &\leqslant (1+2x)\delta +N^2\sqrt{\delta}.
    \end{align}
\end{proof}
\begin{lemma}[Approximately tracial states]\label{lem:approximate_tracial_state}
    In the hypotheses of Lemma~\ref{lem:approx_sync_cor}, let $\rho_0$ be the reduced density matrix of the state $\ket{\Psi}$ in Equation~\eqref{eq:generic_state_2} corresponding to the first system. Then, for all $\ell \in\mathbb N$
    \begin{equation}
        |\tr(\rho_0(WX-XW))|\leqslant 2\ell \sqrt{\delta},
    \end{equation}
    with $X$ being a generic element from the $C^*$-algebra generated by $\{E_\mu^{(0)}\}_{\mu=0,\dots,N-1}$ with norm $\norm{X}\leqslant 1$, and $W$ being a word of order $\ell$ in $\{E_\mu^{(0)}\}_{\mu=0,\dots,N-1}$. Similar properties hold for the other spaces $\mathcal H_k$ for $k\in[1,d]$.
\end{lemma}
\begin{proof}
    The property holds for $\ell=0$, in which case $W=\lambda\mathbb I_{\mathcal H_{0}}$ with $\lambda\in\mathbb C$. Suppose the property holds for $\ell-1$, and take $W=E_{\mu,0}^{(0)}W'$, with $W'$ the word of order $\ell-1$. Then
    \begin{align}
        \tr(\rho_0WX)&=\mathrm{tr}(\rho_0 E_{\mu,0}^{(0)} W' X)=\mel*{\Psi}{ E_{\mu,0}^{(0)} W' X\otimes \bigotimes_{k\in[1,d]}\mathbb I_{\mathcal H_{k}}}{\Psi}\approx_{\sqrt{\delta}} \mel*{\Psi}{W'X\otimes\sum_{\substack{\boldsymbol a\in \{0,1\}^{d-1}\\\sum_{k}a_k=d-r}} \tilde E_{\overline{\boldsymbol{\mu}},\boldsymbol a}^{[1,d]}}{\Psi}\nonumber\\
        &\approx_{\sqrt{\delta}}\mel*{\Psi}{  W' X E_{{\mu},0}^{(0)}\otimes \bigotimes_{k\in[1,d]}\mathbb I_{\mathcal H_{k}}}{\Psi}=\mathrm{tr}(\rho_0  W' XE_{\mu,0}^{(0)})\approx_{2(\ell-1)\sqrt{\delta}}\mathrm{tr}(\rho_0 XE_{\mu,0}^{(0)}W')\nonumber\\
        &=\mathrm{tr}(\rho_0 XW).
    \end{align}
    Applying the principle of induction, we arrive at the desired conclusion.
\end{proof}

\bigskip

We are now ready to proceed with the main proof.

\begin{proof}[Proof of Theorem~\ref{th:self-testing_robust}]
    As in the proof from Ref.~\cite[Th.~6.10]{Mancinska2024}, we must construct an operator whose unique maximal eigenvalue is attained on $\ket{\Psi_{\mathrm{S}}}$. To this end we choose
    \begin{equation}
        \tilde R=\sum_{\mu=0}^{N-1}S_\mu\in\mathcal B(\mathcal H^{\otimes d}),
    \end{equation}
    which, being the sum of projectors, has its maximal eigenvalue equal to $N$. This is attained on $\ket{\Psi_{\mathrm{S}}}$ by Proposition~\ref{prop:stabilizing_slater_state}, and it is the unique eigenvector with eigenvalue $N$ provided that Conjecture~\ref{prop:conjecture} holds. Furthermore, let $\lambda_2<N$ be the second highest eigenvalues of $R$. We fix $\varepsilon >0$, and choose $\varepsilon'>0$ such that
    \begin{align}
        \varepsilon' < (N-\lambda_2) \left(\frac{d!dN}{(d-r)!r!} + 1\right)^{-1} ,\label{eq:definition_varepsilon_prime_1}\\
        d\varepsilon' + \beta \sqrt{(2d+1)\varepsilon'+2d\sqrt{\varepsilon'}+\beta}< \varepsilon,\label{eq:definition_varepsilon_prime_2}
    \end{align}
    where
    \begin{equation}
        \beta:=\sqrt{\left(\frac{2d!dN}{(d-r)!r!}+1\right)\frac{\varepsilon'}{N-\lambda_2}}.
    \end{equation}
    There exists $m\in\mathbb N$ and $\delta'>0$ as in Theorem~\ref{th:gowers_hatami_algebras}, but with $\delta'$ and $\varepsilon'$ replacing $\delta$ and $\varepsilon$. Hence, we set
    \begin{equation}
        \delta=\mathrm{min}\left\{\varepsilon',\left(\frac{\delta'}{N+1}\right)^4,\left(\frac{\delta'}{2m}\right)^2\right\}.
    \end{equation}
    Let $p$ be a $d$-party quantum correlation as in Lemma~\ref{lem:approx_sync_cor}, and $\norm{p-\tilde p}_1\leqslant \delta$. From Lemma~\ref{lem:approx_sum_operators} it follows that
    \begin{equation}
        \bigg\lVert{x\mathbb I_{\mathcal H_k}-\sum_{\mu=0}^{N-1}E_{\mu,0}^{(k)}}\bigg\rVert_{\rho_k}\leqslant C{\delta}^{1/4}
    \end{equation}
    with $C=(1+2x)\sqrt{\delta}+N^2$. But then, according to Eq.~\eqref{eq:definition_varepsilon_prime_1}, we have $\delta\leqslant\varepsilon'<1/d<1$,  and
    \begin{equation}
        C\delta^{1/4}=\sqrt{N^2 + (1+2x)\sqrt{\delta}}\frac{\delta'}{N+1}<\sqrt{N^2 + (1+2N)}\frac{\delta'}{N+1}<\delta',
    \end{equation}
    as $x<N$. Furthermore, based on  Lemma~\ref{lem:approximate_tracial_state}, and the fact that $2m\sqrt{\delta}<\delta'$ we infer that conditions (a), (b) and (c) from Theorem~\ref{th:gowers_hatami_algebras} hold for $\{E_{\mu,0}^{0}\}_{\mu\in[N]}$. Similar result can be shown for all $\{E^{(k)}_{\mu,0}\}_{\mu\in[N]}$ and $k\in[d]$. This implies that for all $k$ there exists an isometry $V_k:\mathcal H_{k}\rightarrow \mathbb C^d\otimes \mathcal H'_k$ such that
    \begin{equation}\label{eq:self-testing_measurement_robust}
        \norm{E_{\mu,0}^{(k)}- V_k^* P_\mu\otimes \mathbb I_{\mathcal H'_k}V_k}_{\rho_k}\leqslant \varepsilon'\quad \mu=0,\dots,N-1.
    \end{equation}
    
    We now proceed to prove the self-testing of the state under the assumption that Conjecture~\ref{prop:conjecture} holds. Specifically, we aim to verify that 
    \begin{equation}
        \ket*{\tilde \Psi}:=\bigotimes_{k=0}^{N-1}V_k \ket{\Psi}\approx_{\varepsilon} \ket{\Psi_{\mathrm S}}\otimes \ket{\psi_{\mathrm{junk}}}.
    \end{equation}
    As a consequence of~\cite[Lemma~6.8]{Mancinska2024} and Equation~\eqref{eq:self-testing_measurement_robust}, we have that, for all $\mu\in[N]^d$:
    \begin{equation}
        |\mel{{\Psi}}{E_{\boldsymbol\mu,\overline{\boldsymbol{0}}}}{{\Psi}}-\mel*{{\tilde \Psi}}{\tilde E_{\boldsymbol\mu,\overline{\boldsymbol{ 0}}}\otimes \bigotimes_{k\in[d]} \mathbb I_{\mathcal H'_k}}{{\tilde\Psi}}|\leqslant d\varepsilon',
    \end{equation}
    with usual notation for $E_{\boldsymbol\mu,{\boldsymbol{a}}}$ and $\tilde E_{\boldsymbol\mu,{\boldsymbol{a}}}$. 
    In particular, if Conjecture~\ref{prop:conjecture} holds, and the state is self-tested, $\ket{\Psi_{\mathrm{S}}}$ is also the only eigenvector of $R$ with eigenvalue $N$. Define an approximate $S_\mu'$ and $R'$, which corresponds to the approximate version of $S_\mu$ and $R$
    \begin{equation}
        S_\mu'=\sum_{\substack{\boldsymbol a\in\times \{0,1\}^{d-1}\\\sum_{k}a_k=d-r}}E_{\overline{\boldsymbol{\mu}}|\boldsymbol a},\quad R=\sum_{\mu=0}^{N-1}S_{\mu}'.
    \end{equation}
    Then
    \begin{equation}
       |\mel{\Psi}{R}{\Psi}-\mel*{\tilde \Psi}{\tilde R\otimes \bigotimes_{k\in[d]}\mathbb I_{\mathcal H_k'}}{\tilde \psi}|\leqslant Nd\frac{d!}{(d-r)!r!}\varepsilon'.
    \end{equation}
    The coefficient $d!/((d-r)!r!)$ is counting the number of summands in $S_\mu$. The inequalities
    \begin{equation}
        \abs{\mel{\Psi}{R}{\Psi}-\mel{\Psi_{\mathrm s}}{\tilde R}{\Psi_{\mathrm s}}}\leqslant \norm{p-\tilde p}_1\leqslant \delta\leqslant \varepsilon',
    \end{equation}
    imply that
    \begin{equation}\label{eq:random_inequality}
        \mel*{\tilde \Psi}{\tilde R\otimes \bigotimes_{k\in[d]}\mathbb I_{\mathcal H_k'}}{\tilde \psi}\geqslant N-\left(dN\frac{d!}{(d-r)!r!}+1\right)\varepsilon'.
    \end{equation}
    Let $Q=\ketbra{\Psi_{\mathrm S}}\otimes \left(\bigotimes_{k\in[d]}\mathbb I_{\mathcal H'_k}\right)$, that is the projection of $\tilde R\otimes \left(\bigotimes_{k\in[d]}\mathbb I_{H_k'}\right)$ onto the eigenspace corresponding to $N$ eigenvalue. From~\eqref{eq:random_inequality} and~\cite[Lemma 6.7]{Mancinska2024}, we have
    \begin{equation}
        \alpha:=\norm*{Q\ket*{\tilde\Psi}}^2\geqslant 1-\left(\frac{d!dN}{(d-r)!r!}+1\right)\frac{\varepsilon'}{N-\lambda_2}>0,
    \end{equation}
    with $\lambda_2$ being the second largest eigenvalue of $R$. In particular, the last inequality is a consequence of~\eqref{eq:definition_varepsilon_prime_1}.
    Let $Q\ket*{\tilde \Psi}=\alpha \ket{\Psi_{\mathrm S}}\otimes \ket{\psi_{\mathrm{junk}}}$, and we finally have
    \begin{equation}\label{eq:second_final_form}
        \norm*{\ket*{\tilde \Psi}- \ket{\Psi_{\mathrm S}}\otimes \ket{\psi_{\mathrm{junk}}}}\leqslant \sqrt{2(1-\alpha)}<\sqrt{2(1-\alpha^2)}=\sqrt{\left(\frac{2d!dN}{(d-r)!r!}+1\right)\frac{\varepsilon'}{N-\lambda_2}}=\beta.
    \end{equation}
    From Equation~\eqref{eq:definition_varepsilon_prime_2}, we infer that
    \begin{equation}
        \bigotimes_{k\in[d]} V_k E_{\mu_k,0}^{(k)} \ket{\Psi}\approx_{\varepsilon} \bigotimes_{k\in[d]} \tilde E_{\mu_k,0}^{(k)}\ket{\Psi_{\mathrm S}},\quad \boldsymbol\mu\in[N]^d,
    \end{equation}
    while, from Equation~\eqref{eq:self-testing_measurement_robust}, together with Lemma 6.8 of~\cite{Mancinska2024} we have:
    \begin{equation}
        \bigg\lVert{ \left(\bigotimes_{k\in[d]}E_{\mu_k,0}^{(k)} \ket{\Psi}- \bigotimes_{k\in[d]} V_k^*\tilde E_{\mu_k,0}^{(k)}V_k\right)\ket{\Psi}\bigg\rVert}\leqslant d\varepsilon'.
    \end{equation}
    Therefore, because $V_k$ are isometries, 
    \begin{equation}
         \bigg\lVert\Bigg(\bigotimes_{k=0}^{d-1}V_k A_{\mu_k,0}^{(k)} \ket{\Psi}- \bigotimes_{k=0}^{d-1}V_kV_k^*\tilde E_{\mu_k,0}^{(k)}V_k\Bigg)\ket{\Psi}\bigg\rVert
         \leqslant d\varepsilon'.
    \end{equation}
    Following Eq.~\eqref{eq:second_final_form}, we obtain
    \begin{equation}\label{eq:this_this}
         \bigg\lVert{\bigotimes_{k\in[d]}V_k E_{\mu_k,0}^{(k)} \ket{\Psi}-\bigotimes_{k\in[d]}V_kV_k^*\tilde E_{\mu_k,0}^{(k)}V_k\ket{\Psi_{\mathrm S}}\otimes\ket{\psi_{\mathrm{junk}}}\bigg\rVert}\leqslant d\varepsilon' +\beta.
    \end{equation} 
    Furthermore, from Lemma~\ref{lem:approx_sync_cor}, we have
    \begin{equation}
    \norm{E_{\mu,0}^{(k)2}-E_{\mu,0}^{(k)}}_\rho\leqslant 2\delta,
    \end{equation}
    from which it follows that
    \begin{equation}
    \bigg|{\bigg\lVert{\bigotimes_{k\in[d]} E_{\mu_k,0}^{(k)}\ket{\Psi}}\bigg\rVert^2-\mel{\Psi}{\bigotimes_{k\in[d]} E_{\mu_k,0}^{(k)}}{\Psi}}\bigg|\leqslant \sum_{k=0}^{d-1}\norm{E_{\mu_k,0}^{(k)}-E_{\mu_k,0}^{(k)2}}_{\rho_k}\leqslant 2d\sqrt{\delta}\leqslant 2d\sqrt{\varepsilon'}.
    \end{equation}
    Thus
    \begin{equation}\label{eq:this}
        \bigg|{\bigg\lVert{\bigotimes_{k\in[d]}V_k E_{\mu_k,0}^{(k)}\ket{\Psi}}\bigg\rVert^2-\bigg\lVert{\bigotimes_{k\in[d]} \tilde E_{\mu_k,0}^{(k)}\ket{\Psi_{\mathrm S}}\otimes\ket{\psi_{\mathrm{junk}}}}\bigg\lVert^2}\bigg|\leqslant \varepsilon' +  2d\sqrt{\varepsilon'}.
    \end{equation}
    Finally, Lemma 6.9 of~\cite{Mancinska2024}, together with~\eqref{eq:this} and~\eqref{eq:this_this}, implies:
    \begin{equation}
        \bigg\lVert{\bigotimes_{k\in[d]}V_k E_{\mu_k,0}^{(k)}\ket{\Psi}-\bigotimes_{k\in[d]} \tilde E_{\mu_k,0}^{(k)}\ket{\Psi_{\mathrm S}}\otimes\ket{\psi_{\mathrm{junk}}}}\bigg\rVert\leqslant d\varepsilon'+\beta+\sqrt{(2d+1)\varepsilon' +  2d\sqrt{\varepsilon'}+\beta}\leqslant \varepsilon 
    \end{equation}
    Similar results hold for all the different $\boldsymbol a\neq \overline{\boldsymbol 0}$.
\end{proof}
\section{Conclusions}

Slater states are antisymmetric quantum states that represent fermionic systems. Constructed as Slater determinants, these states lie in the fermionic (antisymmetric) subspace of the Hilbert space, which notably contains no product states. Consequently, all fermionic states are inherently entangled, even without interaction, due to the indistinguishability of particles~\cite{Benatti2012}. In quantum computing, Slater states are important because they are invariant under global unitaries of the form $U^{\otimes d}$. This symmetry implies protection against collective noise, making them elements of decoherence-free subspaces—a key resource for fault-tolerant computation and quantum error avoidance~\cite{Zanardi1997}. They also play a foundational role in quantum chemistry and condensed matter physics, notably in methods like Hartree–Fock theory and density functional theory, where they approximate complex many-body ground states of many-electron systems.

{Here, we have provided a self-testing scheme for this class of multiqudit states, which, importantly, is optimal in terms of the number of measurements and outcomes as compared to other existing schemes for multiqudit quantum states. In fact, in the case of odd local dimensions, our scheme requires only four binary measurements to be performed by each observer to certify the state. We also extend our result to the case of even local dimensions, however, at the cost of increasing the number of binary measurements to $d$ for a quantum system of local dimension $d$. Moreover, we demonstrate our scheme to be robust to noises and experimental imperfections, which is a key feature for practical implementations of self-testing.}

These findings not only deepen our understanding of multipartite quantum systems but also provide practical tools for the reliable certification of entanglement in realistic settings.

The outlook for this work naturally divides into three directions. First, it would be significant to prove Conjectures~\ref{prop:conjecture} or~\ref{prop:conjecture_3} for arbitrary families of projections that add up to the identity, as considered in Theorem~\ref{th:mancinska}. As discussed in Section~\ref{Section:Algebraic_approach}, a key obstacle lies in the fact that not all irreducible Lie algebras coincide with $\mathcal M_d$, a property that holds for $C^*$-algebras. For instance, the Lie algebra generated by the Pauli matrices $\sigma_1,\sigma_2,\sigma_3$ is $\mathfrak{sl}_2$, the algebra of traceless $2\times 2$ matrices, whereas the $C^*$-algebra they generate is the whole $\mathcal M_2$.
Another interesting direction is to investigate whether the results can be extended to the case of
$N=4$ measurements in even dimensions, analogously to the extension for the maximally entangled bipartite state presented in~\cite{Volcic2024}. {From a more general perspective, it would also be very interesting to explore whether the methodology used in this work can be used to 
design self-testing schemes for other bipartite or multipartite states of arbitrary local dimension.}

\section{Acknowledgments}
We thank Yuming Zhao for insightful discussions. This work is supported by the National Science Centre (Poland) through the SONATA BIS project No. 019/34/E/ST2/00369. This project has received funding from the European Union’s Horizon Europe research and innovation programme under grant agreement No 101080086 NeQST.
\appendix
\section{Permutation and symmetric subspaces}\label{app:symmetric_algebras}
In this Appendix, we review the main properties and relations between the representation of the permutation group $\mathbb{P}_n$ and the symmetric algebra $\mathfrak{A}_f$ on tensor products of Hilbert spaces. See also~\cite{Marconi2025} for a recent review on their applications in quantum information theory.

Fiven a finite-dimensional vector space $\mathcal V$, and its $n$-fold tensor product $\mathcal V^{\otimes n}$, we define the symmetrizer as
\begin{equation}
    \mathcal S(v_{0}\otimes\dots\otimes v_{n-1})=\frac{1}{n!}\sum_{\sigma\in \mathbb P_n} v_{\sigma(0)}\otimes\dots\otimes v_{\sigma(n-1)}
\end{equation}
with $v_0,v_1,\dots v_{n-1}\in\mathcal V$. 
The \textit{symmetric subspace} $\mathfrak A_f$ is then the image of $\mathcal S$, once extended linearly to all elements in $V^{\otimes d}$. One can prove that~\cite[page~242]{Procesi2007}:
\begin{equation}
    \mathfrak A_f=\mathrm{span}\{v^{\otimes n}:v\in\mathcal V\}=\mathrm{span}\{v^{\otimes n}:v\in X\},
\end{equation}
where $X$ is a Zariski-dense subset\footnote{A subset 
$X$ of $V$ is called Zariski-dense if the only polynomial 
$h$ on $X$ that vanishes on all elements of $X$ is the zero polynomial.} of $\mathcal V$.

We now focus on the case $\mathcal V=\mathcal B(\mathcal H)$, where $\mathcal H$ is a finite dimensional Hilbert space of dimension $d$. Here, the symmetric space $\mathfrak A_f$ is a $C^*$-subalgebra of $\mathcal B(\mathcal H^{\otimes n})$, and we call it \textit{symmetric algebra}. The Schur-Weyl Theorem allows to write the structure of this algebra explicitly and connect it with the representations of the permutation group $\mathbb P_n$ on $\mathcal H^{\otimes n}$~\cite{WEYL2016,Landsberg2012}. To be more precise, given a permutation $\sigma$ in $\mathbb P_n$, we can consider the unitary representation $V_\pi$ defined as
\begin{equation}
    V_\sigma \ket{\psi_0\otimes\psi_1\otimes\dots \otimes \psi_{n-1}}=\ket*{\psi_{\sigma(0)}\otimes\psi_{\sigma(1)}\otimes\dots\otimes \psi_{\sigma(n-1)}}.
\end{equation}
The span of these operators $V_\sigma$ defines a $C^*$-algebra called \textit{permutation algebra}:
\begin{equation}
    \mathfrak B=\mathrm{span}\{V_\sigma:\sigma\in\mathbb P_n\}.
\end{equation}

The first result regarding the two spaces is the following.
\begin{theorem}
The permutation group and the symmetric algebra are related by the commutation relations
\begin{equation}
    \mathfrak A_f=\mathfrak B' \qquad \text{and}\qquad  \mathfrak B=\mathfrak A_f'.
\end{equation}
%and
%\begin{equation}
%    \mathfrak B=\mathfrak A_f'.
%\end{equation}
\end{theorem}

As discussed in Section~\ref{sec:algebras}, every finite-dimensional $C^*$-algebra admits a unique decomposition in terms of complete matrix algebras. In the case of symmetric and permutation algebras, this decomposition is obtained from the Young tables, whose complete study is beyond the scope of this work, see~\cite{Landsberg2012} for the details. We restrict our discussion to outlining how to obtain the dimensions of the irreducible representations.

Take a partition $\pi$ of $\{1,2,\dots,n\}$:
\begin{equation}
\pi=(p_0,p_1,\dots,p_{r-1}),\quad p_0\geqslant p_1\geqslant \dots \geqslant p_{r-1},\ \ \sum_{k=0}^{r-1}p_k=n,
\end{equation}

To each partition $\pi=(p_0,p_1,\dots,p_{r-1})$, one can associate an irreducible subrepresentation of the permutation algebra and an irreducible subrepresentation of the symmetric algebra. The dimension of these subrepresentations is determined using the Young diagram $T_\pi$, which represents the paritition $\pi$ in tabular form:

\begin{equation}\label{eq:young_table}
T_{\pi} = \vcenter{\hbox{\ytableausetup{boxsize=1.2cm}
\begin{ytableau}
1 & 2 & 3 & \dots & p_0 \\
p_0+1 & \dots & \dots & p_0+p_1 \\
\vdots & \vdots & \vdots \\
\dots & n
\end{ytableau}}}
\end{equation}\\

The irreducible representation of $\mathfrak B$ associated with the permutation $\pi$ is defined on a Hilbert space $\mathcal H_{\pi,1}$ of dimension:

\begin{equation}
    \mathrm{dim}\mathcal H_{\pi,1}=\frac{n!}{\Pi_{x\in T_\pi}h(x)}.
\end{equation}
Here, $h(x)$ is the hook length of $x$ ($x$ is a box in the Young table), and it is given by the number of boxes to the right of $x$ in the same row, plus the number of boxes below $x$ in the same column, plus one.

The irreducible subrepresentation of $\mathfrak A_f$ acts on a subspace $\mathcal H_{\pi,2}$ of dimension:
    \begin{equation}\label{eq:dimensionS}\dim \mathcal H_{\pi,2} = \prod_{x\in \pi} \frac{n+c(x)}{h(x)},
    \end{equation}
with $c(x)$ being the content of $x$, namely, the number of steps from the diagonal of $x$, with a positive sign if $x$ is over the diagonal and negative one, if it is below.
All together, these results are summarized in the following

\begin{theorem}[Shur-Weyl duality]
    Let $\mathcal H$ be an Hilbert space of dimension $d$. Then, we can decompose $\mathcal H^{\otimes n}$ as
    \begin{equation}
        \mathcal H^{\otimes n}=\bigoplus_{\pi}\mathcal H_{\pi,1}\otimes \mathcal H_{\pi,2}.
    \end{equation}
    The sum is over all possible partitions $\pi$ of $\{1,2,\dots,n\}$ for which the dimension~\eqref{eq:dimensionS} is positive. 
    An element in $X\in\mathfrak A_f$ can be decomposed as
    \begin{equation}
        X=\bigoplus_{\pi}\mathbb I_{{\pi,1}} \otimes X_\pi, 
    \end{equation}
    while an element $Y\in\mathfrak B$ can be decomposed as
    \begin{equation}
        Y=\bigoplus_{\pi}Y_\pi\otimes \mathbb I_{\pi,2}.
    \end{equation}
\end{theorem}
Let us consider two examples of partitions. 
The trivial partition $\pi=\{d\}$, in which case $\mathcal H_{\pi,2}$ is the symmetric or bosonic space, while the dimension of $\mathcal H_{\pi,1}$ is $1$. The partition $\pi=(1,1,\dots,1)$, for which $\mathcal H_{\pi,2}$ is the completely antisymmetric or fermionic space, and again $\dim\mathcal H_{\pi,1}=1$. In the particular case of $d=n$, we have  $\dim \mathcal H_{\pi,2}=1$, and there is only partition $\pi$ for which this happens. This implies that $\ketbra{\Psi_S}$ is the only pure state that commutes with all the operators in $\mathfrak A_f$.

% \rev{Show explicitly why, with an example.}

% \rev{Explain that, in general, the ratio is greater than one, and show that in order to be always one, we need the antisymmetric condition.}

\section{Proof of~\ref{th:conjecture_four_projections}}\label{app:4_proj}
We shall prove that the Lie algebra $\mathscr{LW}$ generated by the four projections given in Section~\ref{sec:4_proj} conincides with $\mathcal M_{2k+1}$. Since the proof involves an explicit computation of the Lie algebra, it is rather lengthy. For clarity, we divide it into three parts. In part I, we evaluate the iterated commutator of $A$ and $B$, thus obtaining a sequence of operators from $\mathscr{LW}$. Using this sequence, in Part II we obtain $\mathscr{LW}_{AB}$, the Lie algebra generated by $A$ and $B$. It turns to be the direct sum of the traceless $2\times 2$ Lie-algebra on each subspace $\ket{2\ell}\oplus\ket{2\ell+1}$, direct sum with $A$. 

Let
\begin{equation}
    \mathfrak{sl}_2=\{x_1\sigma_1+x_2\sigma_2+x_3\sigma_3:x_i\in\mathbb C\}
\end{equation}
be the traceless or special linear $2\times 2$ Lie algebra, with the Pauli matrices given by
\begin{equation}\label{eq:pauli_matrices}
    \sigma_1=\begin{pmatrix}
    0 & 1\\
    1 &0
    \end{pmatrix},\quad
    \sigma_2=\begin{pmatrix}
    0 & -\mathrm i\\
    \mathrm i &0
    \end{pmatrix},\quad
    \sigma_3=\begin{pmatrix}
    1 & 0\\
    0 &-1
    \end{pmatrix}.
\end{equation}
Then, we can express $\mathscr{LW}_{AB}$ as:
\begin{equation}\label{eq:LAB_algebra}
    \mathscr{LW}_{AB}=\mathbb C A+\left(\bigoplus_{\ell=0}^{k-1}\mathfrak{sl}_2\oplus 0 \right).
\end{equation}
In the previous equation, for each $\ell$, $\mathfrak{sl}_2$ has support over $\ket{2\ell}\oplus\ket{2\ell+1}$.
Similarly, it is possible to obtain $\mathscr{LW}_{CD}$, the Lie algebra generated by $C$ and $D$, which possesses a similar structure to $\mathscr{LW}_{AB}$. Specifically, it takes the form:
\begin{equation}\label{eq:LCD_algebra}
\mathscr{LW}_{CD}=\mathbb C C+\left(0\oplus\bigoplus_{\ell=1}^{k}\mathfrak{sl}_2\right),
\end{equation}
with each $\mathfrak{sl}_2$ being the traceless Lie algebra with support over $\ket{2\ell-1}\oplus\ket{2\ell}$. Finally, in part III, we will compute $\mathscr{LW}$ by sticking together $\mathscr{LW}_{AB}$ and $\mathscr{LW}_{CD}$.

\bigskip

\textit{Part I}

We begin with considering the iterated commutator between $A$ and $B$. Using the notation from Section~\ref{sec:4_proj}, take $A$ and $B$ from Equation~\eqref{eq:a_b_projections} in Proposition~\ref{prop:structure_of_projections}. Then
\begin{equation}
    [A,B]=\bigoplus_{\ell=0}^{k-1}[A_\ell,B_\ell]\oplus 0,
\end{equation}
with $A_\ell$ and $B_\ell$ given in Equation~\eqref{eq:operators_A_B}. In the explicit form it reads:
\begin{equation}
    [A_\ell,B_\ell]=
    \frac{\sqrt{x_\ell^{(1)}x_\ell^{(2)}}}{2}\begin{pmatrix}
        0 & x_\ell^{(2)}-x_\ell^{(1)}\\
        x_\ell^{(1)}-x_\ell^{(2)} & 0
    \end{pmatrix}=\frac{w_{\ell}}{2}
    \begin{pmatrix}0 &z_\ell \\ - z_\ell &0
    \end{pmatrix}.
\end{equation}
with
\begin{align}
     z_\ell\equiv x_{\ell}^{(2)}-x_{\ell}^{(1)}&=\frac{4k-2-8\ell}{2k+1}\\
     w_\ell\equiv\sqrt{x_\ell^{(1)}x_\ell^{(2)}}&=\frac{\sqrt{(4\ell+2) (4k-4\ell)}}{2k+1}.
\end{align}
It is convenient to define $Z_\ell^{(1)}=2[A_\ell,B_\ell]$, and $Z^{(1)}=2[A,B]$,
as well as
\begin{equation}
    Z^{(n+1)}=2[Z^{(n)},A]\in\mathscr{LW}_{AB},\quad n\geqslant 2.
\end{equation}
One can easily confirm that for all $n\geqslant 1$ it can be expressed as
\begin{equation}
    Z^{(n)}=\bigoplus_{\ell=0}^{k-1}Z_\ell^{(n)}\oplus 0
\end{equation}
with
\begin{equation}
    Z_\ell^{(n)}=w_\ell
    \begin{pmatrix}
        0 & z_\ell^n \\
        (-z_\ell)^n & 0
    \end{pmatrix}.
\end{equation}
Let $Z^{(0)}=A-B\in\mathscr{LW}_{AB}$, so that the previous expression can be extended to $n=0$. As a result, we can decompose $Z^{(n)}$ as:
\begin{align}
    Z^{(2n)}= \bigoplus_{\ell=0}^{k-1} w_\ell z_{\ell}^{2n} \sigma_1\oplus 0,\quad
    Z^{(2n+1)}=\bigoplus_{\ell=0}^{k-1} \mathrm iw_\ell z_{\ell}^{2n+1} \sigma_2\oplus 0,
\end{align}
with $n\geqslant 0$, and for each $\ell=0,\dots,k-1$ the Pauli matrices~\eqref{eq:pauli_matrices} have support in the subspace $\ket{2\ell}\oplus\ket{2\ell+1}$.

\bigskip

\textit{Part II}

After introducing the operators $Z^{(2n)}$, we prove that $\mathscr{LW}_{AB}$ is given by~\eqref{eq:LAB_algebra}. Let, for all $n\geqslant 0$:
\begin{align}
Z^{(2n,1)}&=z_{k-1}^2 Z^{(2n)}-Z^{(2n+2)}= \bigoplus_{\ell=0}^{k-1}w_\ell z_{\ell}^2(z_{k-1}^2-z_{\ell}^2)\sigma_1\oplus 0\nonumber\\
&=w_0 z_{0}^{2n} (z_{k-1}^2-z_0^2)\sigma_1\oplus w_1 z_{1}^{2n} (z_{k-1}^2-z_1^2)\sigma_1\oplus\dots\oplus  w_0 z_{k-2}^{2n} (z_{k-1}^2-z_{k-2}^2)\sigma_1\oplus 0\sigma_1\oplus 0.
\end{align}
Clearly, every $Z^{(2n,1)}$ belongs to $\mathscr{LW}_{AB}$, as it is a linear combination of two elements in $\mathscr{LW}_{AB}$. Furthermore, the last element corresponding to $\ell=k-1$ vanishes, while all remaining elements are different from $0$,
\begin{equation}
    z_{\ell}^2-z_{\ell'}^2=0\implies \ell=\ell' \textnormal{ or } \ell+\ell'=k-\frac{1}{2},
\end{equation}
and the second condition is not achievable if $\ell$ and $\ell'$ are natural numbers.

 We can repeat the same construction recursively for $j=2,\dots,k-2$ by defining
\begin{equation}
    Z^{(2n,j)}=z_{k-j}^2 Z^{(2n,j-1)}-Z^{(2n+2,j-1)},\quad n\in\mathbb N.
\end{equation}
As a result, $Z^{(2n,j)}$ has all the components $k-j,k-j+1,\dots, k-1$ equal to $0$. In particular, for $j=k-1$, we obtain
\begin{equation}
    Z^{(2n,k-1)}=w_0 z_0^{2n}\prod_{j=1}^{k-1} (z^2_{j}-z_0^2)\sigma_1\oplus 0\oplus\dots\oplus 0.
\end{equation}
By linearity, all operators proportional to $\sigma_1$ on the subspace $\ket{0}\oplus\ket{1}$ and vanishing elsewhere belong to $\mathscr{LW}_{AB}$. Furthermore, taking
\begin{equation}
    Z^{(2n,k-2)}-w_0 z_0^{2n}\prod_{j=2}^{k-1}(z_j^2-z_0^2)\sigma_1\oplus 0\oplus\dots\oplus 0=0\oplus w_0 z_0^{2n}\prod_{j=2}^{k-1}(z_j^2-z_0^2)\sigma_1\oplus\dots\oplus 0,
\end{equation}
we infer that all elements proportional to $\sigma_1$ on $\ket{2}\oplus\ket{3}$ belong to $\mathscr{LW}_{AB}$. Repeating this argument for $Z^{(2n,k-2)},Z^{(2n,k-2)},\dots,Z^{(2n,1)}$, we see that all operators of the form
\begin{equation}\label{eq:sigma_1_AB}
    \bigoplus_{\ell=0}^{k-1} a_\ell \sigma_1\oplus 0,\quad a_0,\dots,a_{k-1}\in\mathbb C
\end{equation}
belong to $\mathscr{LW}_{AB}$. Taking the double iterated commutants of $Z^{(1)}$ 
with~\eqref{eq:sigma_1_AB}, we further have that
\begin{equation}\label{eq:sigma_23_AB}
    \bigoplus_{\ell=0}^{k-1} b_\ell \sigma_2\oplus 0,\quad 
    \bigoplus_{\ell=0}^{k-1} c_\ell \sigma_3\oplus 0,\quad b_0,\dots,b_{k-1},c_0,\dots,c_{k-1}\in\mathbb C
\end{equation}
are all in $\mathscr{LW}_{AB}$, for all sequences $b_\ell,c_\ell$ in $\mathbb C$. In conclusion, the formula in Equation~\eqref{eq:LAB_algebra} holds as the set of $\mathscr{LW}_{AB}$ is closed under commutation.

To prove the fact that $\mathscr{LW}_{CD}$ is the Lie algebra generated by $C$ and $D$, is similar. In this case, we can write

\begin{equation}
    [C_\ell,D_\ell]=
    \frac{\sqrt{y_\ell^{(1)}y_\ell^{(2)}}}{2}\begin{pmatrix}
        0 & y_\ell^{(2)}-y_\ell^{(1)}\\
        y_\ell^{(1)}-y_\ell^{(2)} & 0
    \end{pmatrix}=\frac{w_{\ell}}{2}
    \begin{pmatrix}0 &z_\ell \\ - z_\ell &0
    \end{pmatrix},
\end{equation}
with
\begin{align}
     z_\ell\equiv y_{\ell}^{(2)}-y_{\ell}^{(1)}&=2-\frac{8\ell}{2k+1},\nonumber\\
     w_\ell\equiv\sqrt{y_\ell^{(1)}y_\ell^{(2)}}&=\frac{\sqrt{4\ell (4k+2-4\ell)}}{2k+1},
\end{align}
and the proof can be repeated to obtain~\eqref{eq:LCD_algebra}.
% \begin{equation}
%     z_{\ell}^2-z_{\ell'}^2=0\implies \ell=\ell' \textnormal{ or } \ell+\ell'=\frac{2k+1}{2},
% \end{equation}

\bigskip

\textit{Part III}

Once the Lie algebras generated by $A,B$ and $B,C$ are given, we can study the Lie algebra generated by $A,B,C,D$, and we show that it conincides with $\mathcal M_{2k+1}$.

First, we prove that the elements in the form $E_{i,j}$, with $i<j$, are in $\mathscr{LW}$. Observe that, if $j+1=i$ then $E_{i,i+1}$ is either an element of $\mathscr{LW}_{AB}$ or $\mathscr{LW}_{CD}$. In the general case, $E_{i,j}$ can be seen as an iterated commutator of elements in either of $\mathscr{LW}_{AB}$ or $\mathscr{LW}_{CD}$:
\begin{align}
    E_{i,j}&=[E_{i,i+1},E_{i+1,j}]=[E_{i,i+1},[E_{i+1,i+2},[E_{i+2,j}]]]\nonumber\\
    &=\dots=[E_{i,i+1},[E_{i+1,i+2},[\dots[E_{j-2,j-1},E_{j-1,j}]\dots]]\in \mathscr{LW}.
\end{align}
Similarly, $E_{i,j}\in\mathscr{LW}$ for $i>j$. Regarding the diagonal elements, for all $i<j$:
\begin{equation}
    E_{i,i}-E_{j,j}=E_{i,i}-E_{i+1,i+1}+E_{i+1,i+1}-\dots+E_{j-1,j-1}-E_{j,j}\in\mathscr{LW}.
\end{equation}
In conclusion, the class of traceless operators $\mathfrak{sl}_{2k+1}$ is contained in $\mathscr{LW}$, as a basis of $\mathfrak{sl}_{2k+1}$ is given by $\{E_{i,j}\}_{i\neq j=1}^{2k+1}\cup\{E_0-E_i\}_{i=1}^{2k+1}$. Since also the identity $\mathbb I_{2k+1}$ is in $\mathscr{LW}$~\eqref{eq:sum_4_proj}, we have that $\mathscr{LW}=\mathcal M_{2k+1}$.

\printbibliography
\end{document}